\documentclass[
final
]{dmtcs-episciences}

\usepackage[utf8]{inputenc}
\usepackage{subfigure}

\usepackage{algorithm,algpseudocode}

\usepackage[round]{natbib}

\newtheorem{theorem}{Theorem}
\newtheorem{lemma}{Lemma}

\author[Mohsen Alambardar Meybodi et al.]{Mohsen Alambardar Meybodi\affiliationmark{1}
  \and Abolfazl Poureidi\affiliationmark{2}}
\title[On {[1,2]}-Domination in Interval and Circle Graphs ]{On $[1,2]$-Domination in Interval and Circle Graphs }
\affiliation{
  Department of Applied Mathematics and Computer Science, Faculty of Mathematics and Statistics, University of Isfahan, Isfahan, Iran.\\
  Faculty of Mathematical Sciences, Shahrood University of Technology, Shahrood, Iran.}
\begin{document}
% This is only used if you are compiling for a volume before vol 25
% \publicationdetails{VOL}{2015}{ISS}{NUM}{SUBM}
% This is the new form of collecting the data, starting with vol 25
\publicationdata{vol. 26:3}{2024}{6}{10.46298/dmtcs.13194}{2024-03-08; 2024-03-08; 2024-06-13; 2024-08-22}{2024-09-19}
\maketitle
\begin{abstract}
	\vspace{0.2cm}
  A subset $S$ of vertices in a graph $G=(V, E)$ is a Dominating Set if each vertex in $V(G)\setminus S$ is adjacent to at least one vertex in $S$. Chellali et al. in 2013, by restricting the number of neighbors in $S$ of a vertex outside $S$, introduced the concept of $[1,j]$-dominating set.  A set $D \subseteq V$ of a graph $G = (V, E)$ is called a $[1,j]$-Dominating Set of $G$ if every vertex not in $D$ has at least one neighbor and at most $j$ neighbors in $D$. 
  The Minimum $[1,j]$-Domination problem is the problem of finding the minimum $[1,j]$-dominating set $D$. Given a positive integer $k$ and a graph $G = (V, E)$, the $[1,j]$-Domination Decision problem is to decide whether $G$ has a $[1,j]$-dominating set of cardinality at most $k$. A polynomial-time algorithm was obtained in split graphs for a constant $j$ in contrast to the Dominating Set problem which is NP-hard for split graphs. This result motivates us to investigate the effect of restriction $j$ on the complexity of $[1,j]$-domination problem on various classes of graphs. Although for $j\geq 3$, it has been proved that the minimum of  domination is equal to minimum $[1,j]$-domination in interval graphs, the complexity of finding the minimum $[1,2]$-domination in interval graphs is still outstanding. In this paper, we propose a polynomial-time algorithm for computing a minimum $[1,2]$-dominating set on interval graphs by a dynamic programming technique. Next, on the negative side, we show that the minimum $[1,2]$-dominating set problem on circle graphs is $NP$-complete.
  \end{abstract}
\section{Introduction}
One of the active research areas in graph theory is the study of the domination problem and its variants. Applying specific restrictions on the domination problem leads to various extensions of the domination problem. In independent domination, the additional restriction is that every vertex in dominating set, like $D$, is not adjacent to any other vertex in $D$. In total domination, the additional restriction is that every vertex in $D$ is adjacent to at least one vertex in $D$. In $[1,j]$-domination, the additional restriction is that every vertex not in $D$ is adjacent to at least one and at most $j$ vertices in $D$. Combinatorial and algorithmic results have been widely obtained on the various extensions of domination problems. A discussion of these results can be found in \cite{haynes1998fundamentals}.

\subsection{Definition and notation}\label{Preliminaries}
We begin with some terminology and notation.
\paragraph{Graph}
Let $G=(V, E)$ be a simple graph with vertex set $V$ and edge set $E$. The order of the graph $G$ is defined as $n = \vert V(G)\vert$. The open neighborhood of a vertex $v\in V$, denoted by $N(v)$, is defined as $\{u: \{u, v\} \in E\}$. Also, the closed neighborhood of $v$ is defined as $N[v]=N(v) \cup \{v\}$. We use $N(S)$ and $N[S]$ to denote the open and closed neighborhood $S$, for a set $S\subseteq V$, respectively. That is, $N[S]=\bigcup_{v\in S} N[v]$ and $N(S)=N[S]\setminus S$. The complete bipartite graph $K_{1;3}$ is called a claw, and a graph without a claw as an induced subgraph is called claw-free. 
\paragraph{Domination and {[1,j]-Domination}} A set $S \subseteq V$ is a dominating set if every vertex not in $S$ is adjacent to at least one vertex in $S$. The cardinality of the smallest dominating set, denoted by $\gamma(G)$, is called the domination number. Finding a minimum dominating set was one of the first problems shown  to be NP-hard in~\cite{garey1979computers}. A set $D \subseteq V$ is called a $[1, j]$-dominating set of $G$ if for each $v \in V \setminus D$ we have $1 \leq |N(v) \cap D| \leq j$, i.e. $v$ is adjacent to at least one but not more than $j$ vertices in $D$. Every graph has at least one $[1, j]$-dominating set since the set of all vertices $V$ is itself a $[1, j]$-dominating set. The size of the smallest $[1, j]$-dominating set of $G$ is denoted by $\gamma[1,j](G)$. In the special case where $j=2$, this domination is already known as Quasiperfect Domination. In the decision version of the $[1,j]$-dominating set problem, the input is a graph $G$ and a positive integer $k$ and the objective is to test whether there is a $[1,j]$-dominating set of size at most $k$. 

\paragraph{Interval Graph} Let $I =\{ I_1 , I_2 , \dots , I_n \}$ be a set of intervals on the real line. A graph $G = (V, E)$ is called an interval graph if each of its vertices can be associated with an interval in $I$ and two vertices are adjacent if and only if the corresponding intervals intersect. The set $I$ is called the interval representation of the interval graph $G$. The intervals in $I$ are indexed by increasing the right endpoint. Figure~\ref{Interval} shows an interval graph and its interval representation. 
\begin{figure}[h!]
	\centering
	\includegraphics[scale=.8]{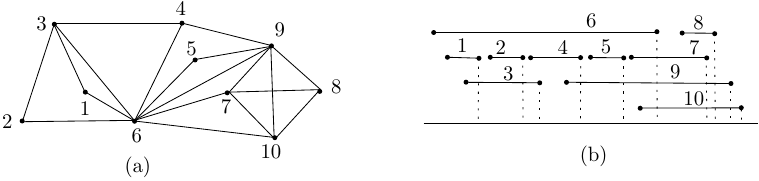}
	\caption{Illustrating (a)  an interval graph $G$ of order $10$, (b) a set of intervals corresponding to the vertex set of  $G$} \label{Interval}
\end{figure}

A linear-time algorithm was proposed by \cite{booth1976testing} to recognize an interval graph and the corresponding intervals of the vertices of an interval graph. The induced subgraphs of interval graphs are also interval graphs which is called hereditary property of the interval graphs~\cite{golumbic1980algorithmic}. A proper interval graph is an interval graph in which no interval is completely contained within another interval. If all the intervals have the same length, the corresponding interval graph is called a unit interval graph.

\paragraph{Circle Graph} A circle graph is a graph $G=(V,E)$ such that there is a one-to-one mapping between the vertices in $V$ and a set $C$ of chords in a circle. Two vertices in $V$ are adjacent if and only if the corresponding chords in $C$ intersect. $C$ is called the chord intersection model for $G$. Equivalently, the vertices of a circle graph can be placed in one-to-one correspondence with intervals such that two vertices are adjacent if and only if the corresponding intervals overlap, but neither contains the other. $I$ is called the interval representation of the circle graph. There exists a polynomial transformation between a representation of a circle graph, either the set of its chords or its interval representation. (See Figure~\ref{circle})

\begin{figure}[!h]
	\centering
	\includegraphics[scale=.7]{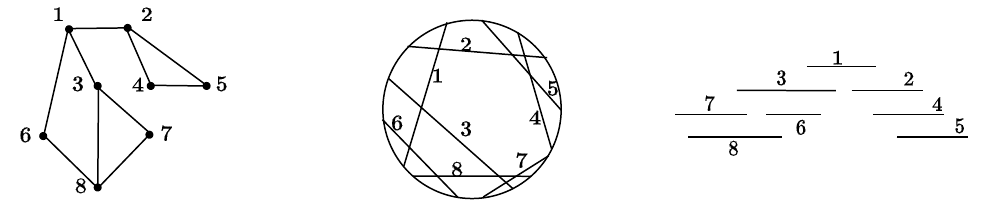}
	\caption{(a) The circle graph $G$ on $8$ vertices (b)The circle representation of $G$ (c)The interval representation of $G$} \label{circle}
\end{figure}

\subsection{Short review of $[1,j]$-Domination}

\cite{chellali20131} introduced the concept of a $[1, j]$-dominating set with name $(1,j)$-set. Although this concept for the constant $j=2$ was already known as quasiperfect domination and defined by~\cite{dejter2009quasiperfect}. Also,  the $[1, j]$-dominating set fits into the general framework of $(\rho, \sigma)$-sets of graphs, which was first introduced in \cite{telle1994complexity}. A $(\rho, \sigma)$-set is a set $D_{(\rho, \sigma)} \subseteq V$ such that for every vertex $v \in V$, we have $|N(v)\cap D_{(\rho, \sigma)}|\in \rho$ if $v \in D_{(\rho, \sigma)}$, and $|N(v)\cap D_{(\rho, \sigma)}|\in \sigma$ if $v \notin D_{(\rho, \sigma)}$. Discussion about the complexity of $[1,j]$-dominating set in various graphs has been done extensively. \cite{chellali20131} raised several open problems. One major open question is in what classes of graphs the  $\gamma[1,j](G)$ is equal to $\gamma(G)$. \cite{chellali20131}  showed that in $P_4$-free graphs, claw-free graphs, and caterpillars $\gamma[1,j](G)$ is equal to $\gamma(G)$. \cite{sharifani2017explicit} showed that $\gamma[1,j](G)=\gamma(G)$ when $G$ is a grid graph. It was proved in \cite{yang20141}, that the value of  $\gamma_{[1,2]}(G)$ in the classes of planar, bipartite, and triangle-free graphs can be equal to the entire set of vertices. So, $\gamma_{[1,2]}$ is not necessarily equal to $\gamma$ in these classes. \cite{etesami2019optimal}  showed that in interval graphs and permutation graphs $\gamma_{[1,3]}(G)=\gamma(G)$. Moreover, it was shown that there exists an interval graph $H$ for which $\gamma(H) < \gamma_{[1,2]}(H)$. Although they proved that for every unit interval graph $G$ we have $\gamma_{[1,2]}(G) = \gamma(G)$, the complexity of finding $\gamma_{[1,2]}$ in non-proper interval graphs remains unsolved. 

Another open question in \cite{chellali20131} was in which graph classes the $[1,j]$-dominating set problem is efficiently solvable. On the negative side, it was also proven that the problem is $NP$-complete even for bipartite graphs\cite{chellali20131}. In ~\cite{bishnu20161}, the $[1, j]$-dominating sets problem was shown to be $NP$-hard even for chordal and planar graphs. In \cite{etesami2019optimal} , the authors showed the complexity of the decision problem of whether the $\gamma_{[1,j]}(G)=\gamma(G)$ is an $NP$ problem. 

On the positive side, a linear-time algorithms proposed for computing $\gamma_{[1,2]}(G)$ when $G$ is a tree  \cite{goharshady20161}, a block graph \cite{sharifani2020linear} and a series-parallel graph \cite{sharifani2017linear}.  For a constant $j$, a polynomial-time algorithm running in roughly $O(n^j p(\log n))$ where $p$ is a polynomial function was obtained for $n$-vertex split graphs~\cite{bishnu20161}. This is in contrast to the orginal dominating set problem which is NP-hard for this class of graphs. Moreover in \cite{meybodi2020parameterized}, a lower bound for computing $[1,j]$-dominating set in the split graph was shown, i.e. there is no algorithm with running time $O(n^{j-\epsilon})$ where $\epsilon>0$ unless $P=NP$. The complexity of $[1,j]$-dominating set problem in various graph classes is summarized in Figure~\ref{Graphs1}. It should be noted that this summary is far from being comprehensive and several other classes of graphs exist. In this paper, we focus only on the circle and interval graphs. 

\begin{figure}[!ht]
	\centering
	\includegraphics[scale=.7]{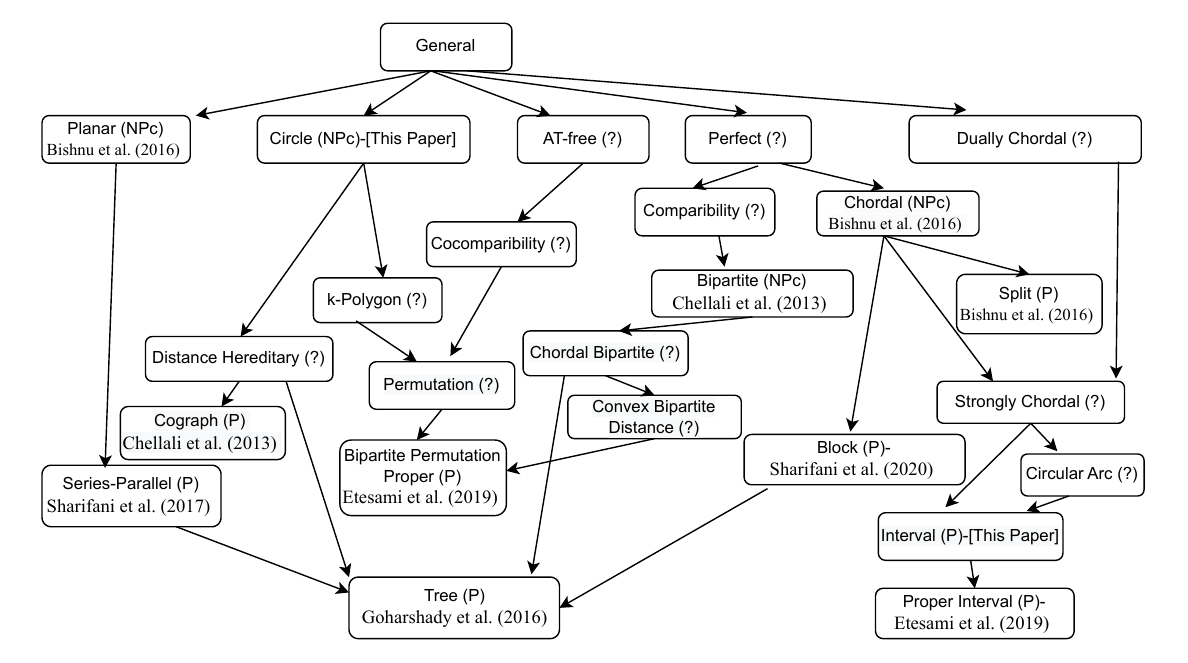}
	\caption{Inclusion relations among well-studied classes of graphs~\cite{tripathi2022complexity} - Complexity of $[1,j]$-dominating set restricted to various classes of graphs in a hierarchy of graphs - $NPc$ indicate that the problem is $NP$-complete, $P$ indicates the problem has a polynomial algorithm, and the question mark indicates the open problem.} \label{Graphs1}
\end{figure}

\section{Algorithm for interval graphs}
\cite{roberts1969indifference} showed that an interval graph is proper if and only if it contains no induced copy of $K_{1;3}$. It was shown in \cite{dejter2009quasiperfect} that for any $K_{1;3}$-free graph $\gamma(G)=\gamma_{[1,2]}(G)$. Thus, the following lemma can be easily concluded.
\\
\begin{lemma}
	For every proper interval graph $G$,  $\gamma(G)=\gamma_{[1,2]}(G)$.
\end{lemma}
\vspace{.5cm}
Another proof for the above lemma was provided in \cite{etesami2019optimal}. The authors also showed that for every interval graph $G$,   $\gamma(G)=\gamma_{[1,3]}(G)$ and there exists a graph $G$ such that $\gamma(G) \le \gamma_{[1,2]}(G)$. However, the complexity of finding a  $[1,2]$-dominating set in interval graphs remains open.

In this section, we propose an $O(n^4)$-time algorithm for computing $\gamma_{[1,2]}(G)$, where $G$ is an interval graph of order $n$. To present and prove the correctness of our algorithm, we need the following lemmas and notation.  
Let $G=(V,E)$ be an interval graph. \cite{ramalingam1988unified} proposed a numbering of the vertices of $G$ and stated the following result. 

\begin{lemma}[\cite{ramalingam1988unified}, Theorem 2.1]\label{theo:RR}
	For  each  interval graph $G=(V,E)$ of order $n$, 
	there is   a  numbering   $( 1, 2,\ldots, n)$ of vertices in $V$  such that    if    $ i k\in E$ and   $i<j<k$, then   $ j k\in E$. This numbering  of vertices  can be computed  in $O(|V|+|E|)$ time.
	
\end{lemma}

In the rest of this section  without loss of generality, we assume that  $G=(V,E)$  is  an   interval graph  of order $n$  with a   numbering  $(1,2,\ldots , n)$ of vertices of $G$ satisfying the condition  of Lemma~\ref{theo:RR}.   Let   $a$   and  $b$  be  integers such that $1\leq   a \leq  b \leq n$.
We  introduce   the following  notation.

\begin{itemize}
	
	\item [$\bullet$]
	$[a, b]=\{ a,a+1,\ldots,b\}$, 
	\item [$\bullet$]
	$(  a,b]=[ a,b]\setminus\{a\}$,
	\item [$\bullet$]
	$[ a, b)=[  a, b]\setminus\{b\}$,  
	\item [$\bullet$]
	$( a, b)=[ a,b]\setminus\{ a,b\}$, 
	\item [$\bullet$]
	$G[a, b]=G[S]$, where $S=[a, b]$, 
	\item [$\bullet$]
	$\mathtt{low} (a)=\min  N[ a] $,  where $N[a]$ is the closed neighborhood of $a$,

	\item [$\bullet$]
	$\mathtt{low} ( a, b)=\min  \{  \mathtt{low} (k):k\in [ a, b]\} $,

	\item [$\bullet$]
	$\mathtt{maxlow}( a)=\max\{\mathtt{low}(k): \mathtt{low}( a)\leq k\leq  a \}$.
\end{itemize}    

Let $i\in [1,n]$, $j\in [1,i)$, $k\in [1,j]$ and $l\in [1,k]$.

\begin{itemize}
	
	\item [$\bullet$]
	$\gamma_{[1,2]} (  i)=\min\{|D|:D$ is a $[1,2]$-dominating  set   of   $G[ 1, i]   \}$,

	\item [$\bullet$]
	$\gamma^0_{[1,2]} (  i)=\min\{|D|:D$ is a $[1,2]$-dominating  set   of   $G[ 1, i] $    (if exists)  such that $i\notin D  \}$,

	\item [$\bullet$]
	$\gamma^0_{[1,2]} ( j, i)=\min\{|D|:D$ is a $[1,2]$-dominating  set   of   $G[ 1, i] $    (if exists)  such that $x\notin D$   for all $x\in [j,i]  \}$,

	\item [$\bullet$]
	$\gamma^0_{[1,2]} ( j,i:k' )=\min\{|D|:D$ is a $[1,2]$-dominating  set   of   $G[ 1, i] $    (if exists)  such that  $k' \in D$  and $x\notin D  $   for all $x\in[j,i]\setminus\{k' \}\}$,   where $k'\in [j,i)$,

	\item [$\bullet$]
	$\gamma^1_{[1,2]} (  i )=\min\{|D|:D$ is a $[1,2]$-dominating  set   of   $G[ 1, i] $   (if exists)  such that $i \in D  \}$,

	\item [$\bullet$]
	$\gamma^1_{[1,2]} ( j,i:i )=\min\{|D|:D$ is a $[1,2]$-dominating  set   of   $G[ 1, i] $   (if exists)  such that $i \in D$   and   $x\notin D$ for all  $x\in [j,i)  \}$,

	\item [$\bullet$]
	$\gamma^1_{[1,2]} ( k,i: i,j )=\min\{|D|:D$ is a $[1,2]$-dominating  set   of   $G[ 1, i] $   (if exists)  such that $i,j\in D$   and   $x\notin D$ for all  $x\in [k,i)\setminus\{j\} \}$,
	
	\item [$\bullet$]
	$\gamma^1_{[1,2]} ( l,i: i,j,k )=\min\{|D|:D$ is a $[1,2]$-dominating  set   of   $G[ 1, i] $   (if exists)  such that $i,j,k\in D$   and   $x\notin D$ for all  $x\in [l,i)\setminus\{j,k\} \}$, 
	
	\item [$\bullet$]
	$\gamma^{11}_{[1,2]} ( l,i: i,j,k )=\min\{|D|:D$ is a $[1,2]$-dominating  set   of   $G[ 1, i] $   (if exists)  such that $i,j,x\in D$  for all $x\in [l,k]$ and   $y\notin D$ for all  $y\in  (k,i)\setminus\{j \} \}$. 
	
\end{itemize}

We  get that   $\mathtt{low}( i )\leq   \mathtt{maxlow}( i)\leq  i$.
An example of an  interval graph is shown in Figure~\ref{fig-IG}.
Also, the values     
$\mathtt{low}( i)$  and $ \mathtt{maxlow}( i)$   
are computed    for each  vertex  $ i$ of the interval graph illustrated in Figure~\ref{fig-IG}.  

\begin{figure}[htb]
	\begin{center}
		\includegraphics[scale=.8]{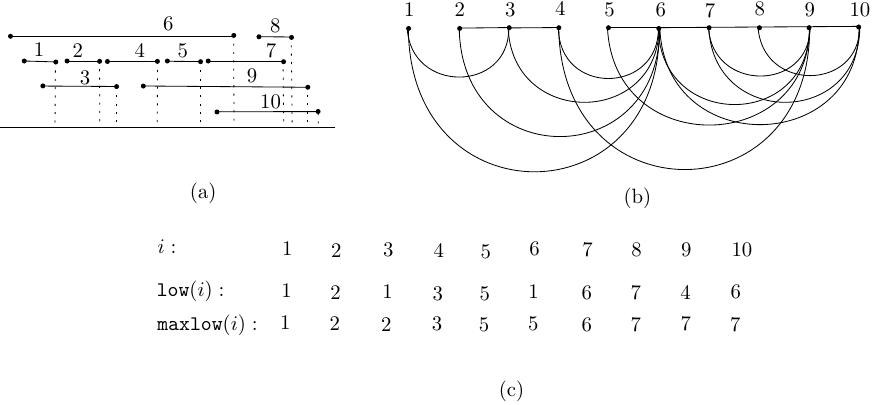}
		\caption{(a) A  set of intervals corresponding to the vertex set of the interval graph $G$; (b)
			a   numbering  $(1,2,\ldots , 10)$ of the vertices of $G$ satisfying the condition  of Corollary~\ref{theo:RR};  and 
			(c)  values $\mathtt{low}(i)$ and  $\mathtt{maxlow}(i)$  for all $ i\in [ 1, 10 ]$}.
		\label{fig-IG}
	\end{center}
\end{figure}

\begin{lemma}[\cite{poureidi2022algorithm}]\label{lem1}
	Let   $ i\in [1,n]$. 
	\begin{itemize}
		\item [(i)]
		The set   $ [\mathtt{maxlow}( i), i]$ is a clique of $G$. 
		\item [(ii)]  
		There is a vertex $k\in [\mathtt{maxlow}( i), i]$ such that   $kl\notin E$ for each  vertex $l\in [ 1,\mathtt{maxlow}( i))$.

	\end{itemize}
\end{lemma}

\begin{lemma} \label{lem2}
	Let  $1\leq i'\leq i\leq n$. 
	\begin{itemize}
		
		\item[(i)]
		$\gamma^0_{[1,2]} (   i)=\gamma^0_{[1,2]} ( i, i)
		$.
		
		\item[(ii)]
		If    $\mathtt{maxlow} ( i)<i'$, then 
		$\gamma^0_{[1,2]} ( i', i)=
		\min\{
		\gamma^1_{[1,2]} ( \mathtt{low} ( j+1,i),j:j),
		\gamma^1_{[1,2]} ( \mathtt{low} ( j+1,i),j: j,k ) 
		: j\in  [\mathtt{maxlow} ( i) ,i'),
		k\in [\mathtt{low} ( j+1,i),j)\} 
		$, otherwise,   $\gamma^0_{[1,2]} ( i', i)$ is not defined.

	\end{itemize} 
\end{lemma}

\begin{proof}
	It is clear that the case  (i) holds. Now, we  prove the case (ii).   Let $ D$ be  a $[1,2]$-dominating  set   of   $G[ 1, i] $  with minimum cardinality,     such that $a\notin D$   for all $a\in [i',i]   $. Thus,    $|D|=\gamma^0_{[1,2]} ( i', i)$. 
	Since $i\notin D$,   at  least one, and at most two of vertices  adjacent to $i$ are in $D$.   If $i'\leq \mathtt{low}(i) $,     then   there is no vertex in $D$ to dominate $i$, that is,      $D$ is not an  $[1,2]$-dominating  set   of   $G[ 1, i] $.  Hence,  $\mathtt{low}(i)<i'$. 
	By case (ii) of Corollary~\ref{lem1},        
	there is a vertex $x\in [\mathtt{maxlow}( i), i]$ such that   $xy\notin E$ for each  vertex $y\in [ 1,\mathtt{maxlow}( i))$.
	Therefore,  if $a\notin D$   for all $a\in   [\mathtt{maxlow}( i), i]$,  then  $D$ is not a $[1,2]$-dominating  set   of   $G[ 1, i] $  and  so  at least one of vertices   of   $  [\mathtt{maxlow}( i), i)$  
	is in $D$.
	Since  $a\notin D$   for all $a\in [i',i]   $, we obtain that $\mathtt{maxlow}(i)<i'$  and  at least one of vertices   of   $  [\mathtt{maxlow}( i), i')$  
	is in $D$.
	Assume that $j\in    [\mathtt{maxlow}( i), i')$ is a vertex in $D$   such that $a\notin D$   for all $a\in (j,i ]$.  We obtain that $D  $ is a $[1,2]$-dominating  set of $G[1,j ] $.    
	Let $b\in     (j,i]$ be a vertex   such that  $\mathtt{low}(b)\leq \mathtt{low}(a)$  for all  $a\in (j,i ]$,   that is,  $\mathtt{low}(b)=\mathtt{low}(j+1,i)$.    See Figure~\ref{fig12}.   Clearly,   $   \mathtt{low}(j+1,i)\leq \mathtt{low}( i)$. 
	Because  $b\notin D$ and $j\in D$, 
	at most one  of vertices of  $[\mathtt{low}(b), j)$   is  in $D$. 
	Hence, either  no vertex of    $[\mathtt{low}(b), j)$   is in $D$,   see Figure~\ref{fig12}(a),   or     exactly   one  vertex of    $[\mathtt{low}(b), j)$   is in $D$,  see Figure~\ref{fig12}(b). In the following we consider these two cases.

	\begin{figure}[htb]
		\begin{center}
			\includegraphics[width=   \textwidth]{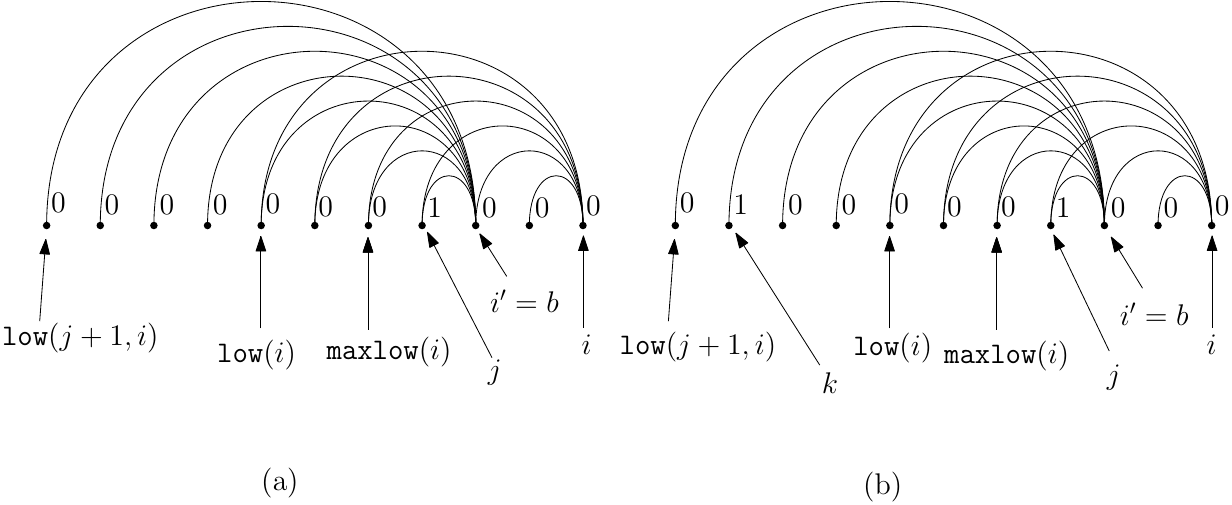}
			\caption{Illustrating   an    $[1,2]$-dominating  set $D$  of   $G[ 1, i] $       such that $a\notin D$   for all $a\in [i',i]   $; (a)    no vertex of    $[ \mathtt{low} ( j+1,i), j)$   is in $D$  and (b)    exactly   one  vertex      $k\in [  \mathtt{low} ( j+1,i), j)$   is in $D$.   Note that a vertex with label 1 is in $D$  and a vertex with label 0 is not in $D$.} 
			\label{fig12}
		\end{center}
	\end{figure}

	\begin{itemize}
		\item [(a)]
		Assume that  no vertex of    $[\mathtt{low}(b), j)$   is in $D$. Hence, $D$ is an  
		$[1,2]$-dominating  set   of   $G[ 1, j] $      such that $j \in D$   and   $a\notin D$ for all  $a\in [\mathtt{low}(b),j)   $  and so  $|D|\leq \gamma^1_{[1,2]} ( \mathtt{low} ( j+1,i),j:j)$. 
		\item [(b)]
		Assume that    exactly   one  vertex      $k\in [\mathtt{low}(b), j)$   is in $D$.   Hence, 
		$D$ is an  
		$[1,2]$-dominating  set   of   $G[ 1, j] $      such that $j,k \in D$   and   $a\notin D$ for all  $a\in [\mathtt{low}(b),j) \setminus\{k\} $  and so 
		$|D|\leq \gamma^1_{[1,2]} ( \mathtt{low} ( j+1,i),j:j,k)$. 
		
	\end{itemize}

	\noindent It follows   from (a)--(b)  that   $ \gamma^0_{[1,2]} ( i', i)=|D| \leq 
	\min\{
	\gamma^1_{[1,2]} ( \mathtt{low} ( j+1,i),j:j),
	\gamma^1_{[1,2]} ( \mathtt{low} ( j+1,i),j: j,k ) 
	: j\in  [\mathtt{maxlow} ( i) ,i'),
	k\in [\mathtt{low} ( j+1,i),j)\} 
	$.

	Conversely,   let  $j\in    [\mathtt{maxlow}( i), i')$.  Assume that $S$ is an  $[1,2]$-dominating  set   of   $G[ 1, j] $ with minimum cardinality      such that $j\in S$   and   $a\notin S$ for all  $a\in [ \mathtt{low} ( j+1,i),j)  $. So,    $|S|=\gamma^1_{[1,2]} (  \mathtt{low} ( j+1,i),j:j )$.  Since   $j\in    [\mathtt{maxlow}( i), i') $ and    $a\notin S$ for all  $a\in [ \mathtt{low} ( j+1,i),j)  $,   each vertex of $[i',i]$ is  exactly  dominated by   $j$. 
	We  obtain that $S$ is an  $[1,2]$-dominating  set   of   $G[ 1, i] $  such that $a\notin S$   for all $a\in [i',i]$,  that is,    $ |S|\leq   \gamma^0_{[1,2]} ( i', i)$.   
	Let $k\in   [\mathtt{low} ( j+1,i),j)$. 
	Assume that $S'$ is an  $[1,2]$-dominating  set   of   $G[ 1, j] $ with minimum cardinality      such that $j,k\in S'$   and   $a\notin S'$ for all  $a\in [ \mathtt{low} ( j+1,i),j) \setminus\{k\} $. So,    $|S'|=\gamma^1_{[1,2]} (  \mathtt{low} ( j+1,i),j:j,k )$.  Since   $j\in    [\mathtt{maxlow}( i), i') $, $k\in   [\mathtt{low} ( j+1,i),j)\cap S'$ and    $a\notin S'$ for all  $a\in [ \mathtt{low} ( j+1,i),j) \setminus\{k\} $,   each vertex of $[i',i]$ is dominated by at least one and at most two vertices of $S'$. 
	We  obtain that $S'$ is an  $[1,2]$-dominating  set   of   $G[ 1, i] $  such that $a\notin S'$   for all $a\in [i',i]$,  that is,    $ |S'|\leq   \gamma^0_{[1,2]} ( i', i)$.   
	Hence,    $
	\min\{
	\gamma^1_{[1,2]} ( \mathtt{low} ( j+1,i),j:j)=|S|,
	\gamma^1_{[1,2]} ( \mathtt{low} ( j+1,i),j: j,k )=|S'| 
	: j\in  [\mathtt{maxlow} ( i) ,i'),
	k\in [\mathtt{low} ( j+1,i),j)\} \leq  \gamma^0_{[1,2]} ( i', i)
	$. 
	This completes the proofs  of the case (ii) and      the lemma. 
\end{proof}

\begin{lemma} \label{lem3}
	Let   $i\in[1,n]$,   $j\in[1,i)$,  $k\in  [j,i)$   and   $b=\min\{\mathtt{low} ( k+1,i),j\}$. 
	If $ \mathtt{maxlow}(i)\leq k $, then $ \gamma^0_{[1,2]} ( j,i:k) =
	\min\{
	\gamma^1_{[1,2]} (b,k:k),
	\gamma^1_{[1,2]} ( b,k:  k,l ) 
	: 
	l\in [b,j)\} 
	$,   otherwise,      $\gamma^0_{[1,2]} ( j,i:k)$ is not defined.  
\end{lemma}

\begin{proof}
	Let $ D$ be  a $[1,2]$-dominating  set   of   $G[ 1, i] $  with minimum cardinality   such that    $k \in D$  and $a\notin D  $   for all $a\in[j,i]\setminus\{k \}$. So,    $|D|=\gamma^0_{[1,2]} ( j, i:k)$. 
	Since $a\notin D$  for all $a\in (k,i]$,  it obtains that $D$ is an  $[1,2]$-dominating  set   of   $G[ 1, k] $. 
	By the case (ii) of Corollary~\ref{lem1},        
	there is a vertex $x\in [\mathtt{maxlow}( i), i]$ such that   $xy\notin E$ for each  vertex $y\in [ 1,\mathtt{maxlow}( i))$.
	Thus,  if    $k<\mathtt{maxlow}(i)$,     then  $D$ is not an  $[1,2]$-dominating  set   of   $G[ 1, i] $,  that is,   $\gamma^0_{[1,2]} ( j,i:k)$ is not defined.   Assume  that   $ \mathtt{maxlow}(i)\leq k $.  
	Let $c\in     (k,i]$ be a vertex   such that  $\mathtt{low}(c)\leq \mathtt{low}(a)$  for all  $a\in (k,i ]$,   that is,  $\mathtt{low}(c)=\mathtt{low}(k+1,i)$.    See Figure~\ref{fig13}.    
	Recall that  $b=\min\{\mathtt{low} ( k+1,i),j\}$  and  $a\notin D  $   for all $a\in[j,i]\setminus\{k \}$.   If  $ \mathtt{low} ( k+1,i)< j $, that is $b=\mathtt{low} ( k+1,i)$,  then 
	because  $c\notin D$ and $k\in D$, 
	at most one  of vertices of  $[b,j)$   is  in $D$, otherwise, $b=j$ and   no vertex of    $[b, k)$   is in $D$.    
	Hence, either  no vertex of    $[b, k)$   is in $D$,   see Figure~\ref{fig13}(a),   or     exactly   one  vertex of    $[b, j)$   is in $D$  (if  $b< j$),  see Figure~\ref{fig13}(b). In the following we consider these two cases.

	\begin{figure}[htb]
		\begin{center}
			\includegraphics[width=   \textwidth]{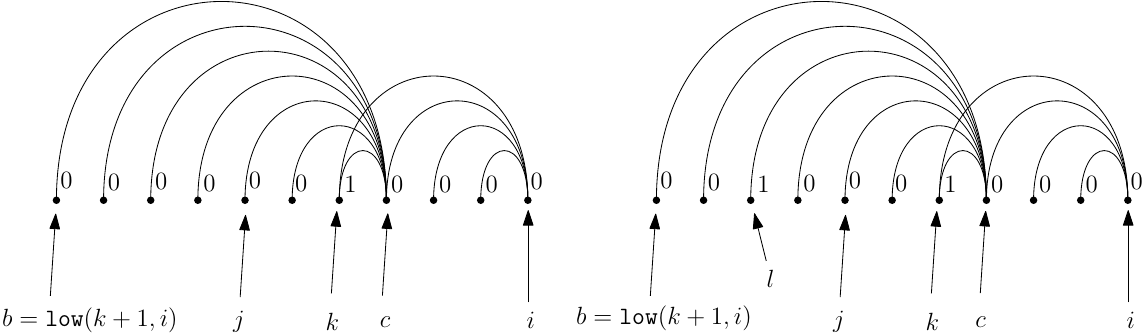}
			\caption{Illustrating   an    $[1,2]$-dominating  set $D$  of   $G[ 1, i] $    
				such that   $k \in D  $  and $a\notin D  $   for all $a\in[j,i]\setminus\{k \}$;  (a)    no vertex of    $[ b, k)$   is in $D$  and (b)    exactly   one  vertex      $l\in [  b, j)$   is in $D$.   Note that a vertex with label 1 is in $D$  and  a vertex with label 0  is  not in $D$.} 
			\label{fig13}
		\end{center}
	\end{figure}

	\begin{itemize}
		\item [(a)]  Assume that   no vertex of    $[b, k)$   is in $D$.
		Hence, $D$ is an  
		$[1,2]$-dominating  set   of   $G[ 1, k] $      such that $k \in D$   and   $a\notin D$ for all  $a\in [b,k)   $  and so  $|D|\leq \gamma^1_{[1,2]} ( b,k:k)$. 
		
		\item   
		[(b)]    
		Assume that    exactly   one  vertex of    $[b, j)$   is in $D$  (if  $b<j$).  Hence,     $D$ is an  
		$[1,2]$-dominating  set   of   $G[ 1, k] $      such that $l,k \in D$, $l\in [b, j)$     and   $a\notin D$ for all  $a\in [b,k) \setminus\{l\} $  and so 
		$|D|\leq \gamma^1_{[1,2]} (b,k:k,l)$. 
		
	\end{itemize}

	\noindent  It follows from (a)--(b) that    $ \gamma^0_{[1,2]} ( j,i:k) =|D|\leq 
	\min\{
	\gamma^1_{[1,2]} (b,k:k),
	\gamma^1_{[1,2]} ( b,k:  k,l ) 
	: 
	l\in [b,j)\} 
	$.

	Conversely,        assume that $S$ is an  $[1,2]$-dominating  set   of   $G[ 1, k] $ with minimum cardinality      such that $k\in S$   and   $a\notin S$ for all  $a\in [ b,k)  $. So,    $|S|=\gamma^1_{[1,2]} (  b,k:k)$.  
	Recall that  $b=\min\{\mathtt{low} ( k+1,i),j\}$. 
	Since   $k\in    [\mathtt{maxlow}( i), i ) $ and    $a\notin S$ for all  $a\in [ b,k)  $,   each vertex of $ (k,i]$ is  exactly  dominated by $k$. 
	We  obtain that $S$ is an  $[1,2]$-dominating  set   of   $G[ 1, i] $  such that   $k\in S$  and $a\notin S$   for all $a\in [j,i]\setminus\{k\}$,  that is,    $ |S|\leq   \gamma^0_{[1,2]} ( j, i:k)$.   
	
	Let $l\in   [b,j)$. 
	Assume that $S'$ is an  $[1,2]$-dominating  set   of   $G[ 1,k] $ with minimum cardinality      such that $l,k\in S'$   and   $a\notin S'$ for all  $a\in [ b,k) \setminus\{l\} $. So,    $|S'|=\gamma^1_{[1,2]} (  b,k:k,l )$.  Since    $l\in   [b,j)$,  $k\in    [\mathtt{maxlow}( i), i ) $   and    $a\notin S'$ for all  $a\in [ b,i) \setminus\{k,l\} $,   each vertex of $ (k,i]$ is dominated by at least one and at most two vertices of $S'$. 
	We  obtain that $S'$ is an  $[1,2]$-dominating  set   of   $G[ 1, i] $  such that   $k\in S'$  and $a\notin S'$   for all $a\in [j,i]\setminus\{k\}$,  that is,    $ |S'|\leq   \gamma^0_{[1,2]} ( j, i:k)$.   
	Hence,    $ \gamma^0_{[1,2]} ( j,i:k) =
	\min\{
	\gamma^1_{[1,2]} (b,k:k)=|S|,
	\gamma^1_{[1,2]} ( b,k:  k,l )=|S'| 
	: 
	l\in [b,j)\} 
	$
	This completes the proof of     the lemma.       
\end{proof}

\begin{lemma} \label{lem24}  
	Let $D$ be an     $[1,2]$-dominating  set   of   $G[ 1,i] $ with minimum cardinality      such that $i\in[1,n]$,  $i \in D$,
	$i'=\max D\cap [1,i)$,    $i''=\max D\cap [1,i')$  and   $\mathtt{low}(i)< i''$.  If both $i'$ and $i''$ exist,  then  $x\in D$   for all $x\in [\mathtt{low}(i), i'') $. 
\end{lemma}

\begin{proof}
	Suppose  that $j=max  [\mathtt{low}(i),i'')\setminus D$.  See Figure \ref{fig15}.
	Since $j\notin D$  and $j$ is adjacent to $i$, the vertex $j$   is adjacent to at most one vertex of $S=D\cap (j,i)$. If $j$ is adjacent to no vertex of   $S $,    then $D'=D\setminus S$   is an  $[1,2]$-dominating  set   of   $G[ 1,i] $ such that $i\in D'$, a contradiction.
	Let     $j$ be  adjacent to exactly  one    vertex    $k\in S $.  We obtain that  $D''=D \setminus S\cup\{k\} $ is an  $[1,2]$-dominating  set   of   $G[ 1,i] $ such that $i\in D'$, a contradiction.
	This completes the proof of the lemma.     
\end{proof}

\begin{figure}[htb]
	\begin{center}
		\includegraphics[width= 0.4  \textwidth]{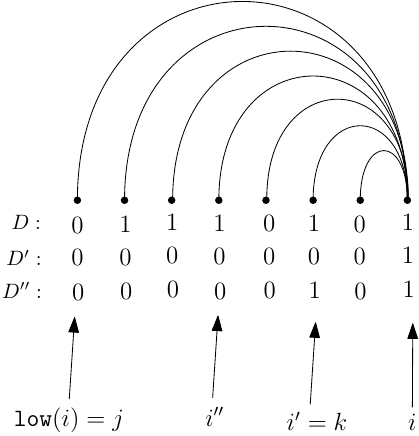}
		\caption{Illustrating   an    $[1,2]$-dominating  set $D$  of   $G[ 1, i] $    
			such that   $i \in D  $.    Note that a vertex with label 1 is in $D$  and  a vertex with label 0  is  not in $D$.} 
		\label{fig15}
	\end{center}
\end{figure}

\begin{lemma} \label{lem14}
	Let    $b=\mathtt{low} ( i)$. 
	\begin{itemize}
		\item [(i).]
		If  $b=i$, then  $\gamma^1_{[1,2]} (i)=\gamma_{[1,2]} (i-1)+1$.
		
		\item [(ii).]
		If  $b<i$,  
		then   $\gamma^1_{[1,2]} (i)=\min\{ \gamma^1_{[1,2]} (b ,i:i), \gamma^1_{[1,2]} (b ,i:i,j) ,\gamma^1_{[1,2]} (b ,i:i,j',k'),\gamma^{11}_{[1,2]} (b ,i:i,j',k')   : b\leq  j<i, b\leq k'< j'<i \} $.
	\end{itemize}
\end{lemma} 
\begin{proof}
	The proof of the case (i) is clear. Assume that  $\mathtt{low} ( i)<i$.  
	Let $D$ be an     $[1,2]$-dominating  set   of   $G[ 1,i] $ with minimum cardinality      such that $i \in D$. So, $|D|=\gamma_{[1,2]}^1(i)$. We obtain that    all vertices of  $[b,i)$ are not in $D$,      exactly one vertex of $[b,i)$   is in $D$,   exactly two vertices  of $[b,i)$   are  in $D$  or     at least  three  vertices  of $[b,i)$   are  in $D$.  If  all vertices of  $[b,i)$ are not in $D$, then
	\begin{eqnarray}\label{equ11}
		\gamma_{[1,2]}^1(b,i:i)\leq |D|=\gamma_{[1,2]}^1( i).
	\end{eqnarray}
	
	\noindent  If  exactly one vertex   $j\in [b,i)$   is in $D$,  then
	\begin{eqnarray}\label{equ12}
		\gamma_{[1,2]}^1(b,i:i,j)\leq |D|=\gamma_{[1,2]}^1( i).
	\end{eqnarray}
	
	\noindent  If  exactly two vertices   $j,k\in [b,i)$   are in $D$ such that $k<j$,  then
	\begin{eqnarray}\label{equ13}
		\gamma_{[1,2]}^1(b,i:i,j,k)\leq |D|=\gamma_{[1,2]}^1( i).
	\end{eqnarray}

	\noindent  Assume that     at least  three  vertices   of   $   [b,i)$    are in $D$  and let $i'=\max D\cap [1,i)$  and    $i''=\max D\cap [1,i')$.   
	By Lemma \ref{lem24}, $x\in D$   for all $x\in [b, i'') $. 
	Hence, 
	\begin{eqnarray}\label{equ14}
		\gamma_{[1,2]}^{11}(b,i:i,i',i'')\leq |D|=\gamma_{[1,2]}^1( i).
	\end{eqnarray}

	Clearly,  $\gamma^1_{[1,2]} (i)\leq \min\{ \gamma^1_{[1,2]} (b ,i:i), \gamma^1_{[1,2]} (b ,i:i,j) ,\gamma^1_{[1,2]} (b ,i:i,j',k'),\gamma^{11}_{[1,2]} (b ,i:i,j',k')   : b\leq  j<i, b\leq k'< j'<i \} $.
	This, together with inequalities (\ref{equ11})--(\ref{equ14}),  completes the proof of the lemma. 
\end{proof}

\begin{lemma} \label{lem4}
	Let $i\in[1,n]$,   $1<b= \mathtt{low} ( i) $  and $a=\max\{j,b\}$. 
	\begin{itemize}
		\item[(i)]
		If $ j\leq  \mathtt{low} ( b-1  )$, then  
		$\gamma^1_{[1,2]} ( j,i:i)$   is not defined.
		
		\item[(ii)]
		If $   \mathtt{low} ( b-1  ) <j<  \mathtt{low} ( b ,i-1 )$, then  
		$\gamma^1_{[1,2]} ( j,i:i)=\gamma^0_{[1,2]} ( j,b-1)+1$.

		\item[(iii)]
		If $  \mathtt{low} ( b ,i-1 )\leq  j  $, then
		$\gamma^1_{[1,2]} ( j,i:i)=
		\min\{
		\gamma^0_{[1,2]} ( \mathtt{low} (c ) , j-1 ),
		\gamma^0_{[1,2]} ( \mathtt{low} (c ),j-1:  k ) ,
		\gamma^1_{[1,2]} ( \mathtt{low} ( c ),j-1: j-1 ) 
		: 
		k\in [\mathtt{low} (c ),j-1)\}+1 
		$, where  $c\in  [a,i-1) $  such that   $  \mathtt{low} ( c )\leq  \mathtt{low} ( x  )$ for   all $x\in [a,i-1) $.

	\end{itemize} 
\end{lemma}

\begin{proof}
	Let $D$  be an      $[1,2]$-dominating  set   of   $G[ 1,i] $ with minimum cardinality      such that $i \in D$   and   $x\notin D$ for all  $x\in [j,i)   $.  So,    $|D|=\gamma^1_{[1,2]} (  j,i:i)$.  If   $ j\leq  \mathtt{low} ( b-1  )$, then there is no vertex  in $  D$ to dominate $b-1$ and so $D$ is not a $[1,2]$-dominating  set   of   $G[ 1,i] $, that is,   $ \gamma^1_{[1,2]} (  j,i:i)$ is not  defined. See Figure \ref{fig14}(a).  This proves the case (i).  In the rest of the proof, we assume that  $    \mathtt{low} ( b-1  )<j$. Recall that $a=\max\{j,b\}$. 
	Let $c\in  [a,i-1) $ be a vertex such that   $  \mathtt{low} ( c )\leq  \mathtt{low} ( x  )$ for   all $x\in [a,i-1) $, that  is,   $  \mathtt{low} ( c )= \mathtt{low} ( a ,i-1  )$.  
	We  distinguish two cases depending on   $j<    \mathtt{low} (c  )$ or   $    \mathtt{low} (c  )\leq j$. In the following we consider these cases.

	\begin{figure}[htb]
		\begin{center}
			\includegraphics[width=  \textwidth]{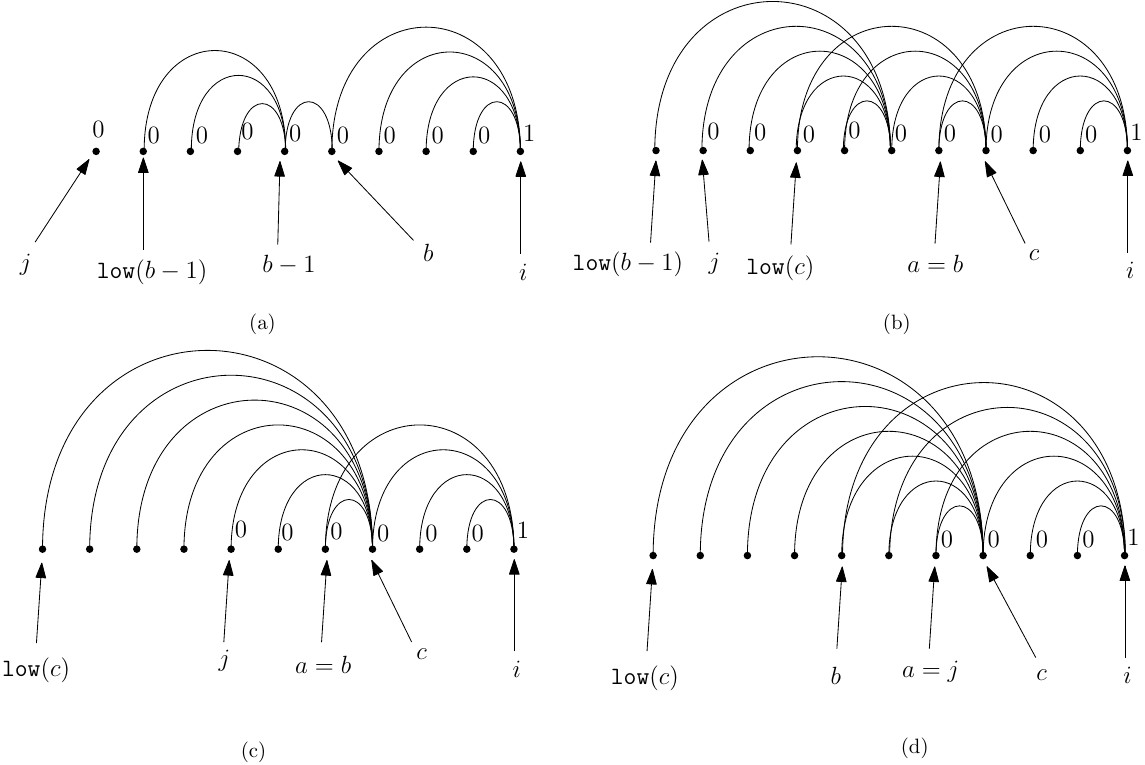}
			\caption{ Illustrating   an    $[1,2]$-dominating  set $D$  of   $G[ 1, i] $    
				such that   $i \in D  $  and $x\notin D  $   for all $x\in[j,i)$;  (a) 
				$ j\leq  \mathtt{low} ( b-1  )$,  (b)  $   \mathtt{low} ( b-1  ) <j<  \mathtt{low} ( c )$, (c)   $  \mathtt{low} (c )\leq j\leq b=a\leq c<i $, and (d)     $  \mathtt{low} (c )\leq b< j=a\leq c<i $,  
				where      $b= \mathtt{low} ( i) $,  $a=\max\{j,b\}$ and  $\mathtt{low}(c) = \mathtt{low} (b,i-1 )$. 
				Note that a vertex with label 1 is in $D$  and  a vertex with label 0  is  not in $D$.} 
			\label{fig14}
		\end{center}
	\end{figure}

	\begin{itemize}
		\item [$\bullet$]   
		Assume that   $j<    \mathtt{low} (c )$. 
		If $b\leq j$, that is, $a=j$, then $ \mathtt{low} (c )=   \mathtt{low} ( j-1,i-1  )\leq   \mathtt{low} ( j )\leq j$,  a contradiction.   Hence,        $j<b$.  See Figure \ref{fig14}(b).  Recall that  $x\notin D$ for all $x\in [j,i)$.  Let $y\in [j,b-1]$.  Since $y\notin D$ and $y$ is not adjacent to $i$,   there  are at least one and at most two vertices  in $D$   that are adjacent to $y$,  that is,  $1\leq |[1,j)\cap D|\leq 2$.  Let $D'=D\setminus\{i\}$.  We get that $D'$ is an 
		$[1,2]$-dominating  set   of   $G[ 1,b-1] $    such that   $x\notin D'$ for all  $x\in [j,b-1]   $.  Thus,  $ \gamma^0_{[1,2]} ( j,b-1)\leq |D'|=\gamma^1_{[1,2]} ( j,i:i)-1$.

		Conversely, assume that $S$ is a $[1,2]$-dominating  set   of   $G[ 1,b-1] $ with minimum cardinality      such that    $x\notin S$ for all  $x\in [j,b-1]   $.  So,    $|S|=\gamma^0_{[1,2]} (  j,b-1)$. 
		We  see that $S'=S\cup\{i\}$ is a $[1,2]$-dominating  set   of   $G[ 1,i] $      such that $i\in S'$ and     $x\notin D$ for all  $x\in [j,i-1]   $.   So, 
		$\gamma^1_{[1,2]} ( j,i:i)\leq |S'|=\gamma^0_{[1,2]} (  j,b-1)+1$. This completes the proof of the case (ii). 
		
		\item [$\bullet$] 
		Assume that   $    \mathtt{low} (c )\leq j$. 
		If $j\leq b$,   then $  \mathtt{low} (c )\leq j\leq b=a\leq c<i $,   
		see Figure \ref{fig14}(c),  otherwise,  $  \mathtt{low} (c )\leq b< j=a\leq c<i $, see Figure \ref{fig14}(d). 
		Since $a=\max\{j,b\}$, we obtain    that  $c$ is adjacent to $i$, where $i\in D$. 
		Since   $x\notin D$ for all  $x\in [j,i)   $, $c\notin D$ and so $c$ is adjacent to at most one vertex of $D\setminus\{i\}$, that is, 
		either  (a) $x\notin D$ for all $x\in [ \mathtt{low} (c ),j-1] $ or  (b)    exactly one vertex of  $  [ \mathtt{low} (c ),j-1] $ is in $D$. In the following we consider these cases. 
		
		\begin{itemize}
			
			\item [(a)]
			
			Assume that  $x\notin D$ for all $x\in [ \mathtt{low} (c ),j-1] $. We obtain that  $D'=D\setminus\{i\}$  is a $[1,2]$-dominating  set   of   $G[ 1,j-1] $  such that   $x\notin D'$   for all $x\in   [ \mathtt{low} (c ),j-1] $.  Hence, 
			
			\begin{eqnarray}\label{equ1}
				\gamma_{1,2}^0(\mathtt{low} (c ),j-1)\leq |D'|=\gamma_{1,2}^1(j,i:i)-1.
			\end{eqnarray} 
			
			\item [(b)]
			Assume that  exactly one vertex   $ k\in  [ \mathtt{low} (c ),j-1] $ is in $D$. 
			If $k<j-1$, then   $D'=D\setminus\{i\}$  is a $[1,2]$-dominating  set   of   $G[ 1,j-1] $  such that  $k\in D'$  and  $x\notin D'$   for all $x\in   [ \mathtt{low} (c ),j-1] $.  Hence, 
			\begin{eqnarray}  
				\gamma_{1,2}^0(\mathtt{low} (c ),j-1:k)\leq |D'|=\gamma_{1,2}^1(j,i:i)-1.
			\end{eqnarray}
			If $k=j-1$, then   $D'=D\setminus\{i\}$  is a $[1,2]$-dominating  set   of   $G[ 1,j-1] $  such that  $j-1\in D'$  and  $x\notin D'$   for all $x\in   [ \mathtt{low} (c ),j-1] $.  Hence, 
			\begin{eqnarray}   
				\gamma_{1,2}^1(\mathtt{low} (c ),j-1:j-1)\leq |D'|=\gamma_{1,2}^1(j,i:i)-1.
			\end{eqnarray}
			
		\end{itemize}
		
		Conversely, assume that $S_1$ is an  $[1,2]$-dominating  set   of   $G[ 1,j-1] $ with minimum cardinality  such that  $x\notin S_1$ for all $x\in [ \mathtt{low} (c ),j-1] $. So, $|S_1|=\gamma_{[1,2]}^0(\mathtt{low} (c ),j-1)$.  Because 
		$  \mathtt{low} ( c )= \mathtt{low} ( a ,i-1  )$,  each vertex of  $  [b,i]$ is adjacent to no vertex of   $S$. Hence, $S'_1=S_1\cup\{i\} $ is an
		$[1,2]$-dominating  set   of   $G[ 1,i] $   such that  $i\in S'_1$ and  $x\notin S_1$ for all $x\in [ j-1,i) $. So, 
		\begin{eqnarray}   
			\gamma_{1,2}^1(j,i:i)\leq |S'_1|=\gamma_{[1,2]}^0(\mathtt{low} (c ),j-1)+1.
		\end{eqnarray}
		
		Assume that $S_2$ is an  $[1,2]$-dominating  set   of   $G[ 1,j-1] $ with minimum cardinality  such that $k\in S_2$ for some $k\in     [ \mathtt{low} (c ),j-1)$ and  $x\notin S_2$ for all $x\in [ \mathtt{low} (c ),j-1] \setminus\{k\}$. So, $|S_2|=\gamma_{[1,2]}^0(\mathtt{low} (c ),j-1:k)$.  Because 
		$  \mathtt{low} ( c )= \mathtt{low} ( a ,i-1  )$,  each vertex of  $  [b,i]$ is adjacent to at most one vertex    $k$. Hence, $S'_2=S_2\cup\{i\} $ is an
		$[1,2]$-dominating  set   of   $G[ 1,i] $   such that  $i\in S'_2$ and  $x\notin S_2$ for all $x\in [ j-1,i) $. So, 
		\begin{eqnarray}   
			\gamma_{1,2}^1(j,i:i)\leq |S'_2|=\gamma_{[1,2]}^0(\mathtt{low} (c ),j-1:k)+1.
		\end{eqnarray}
		
		Similarly, we obtain that  
		\begin{eqnarray} \label{equ6}  
			\gamma_{1,2}^1(j,i:i)\leq |S'_2|=\gamma_{[1,2]}^1(\mathtt{low} (c ),j-1:j-1)+1.
		\end{eqnarray}
		
		The proof of the  case (iii)   follows from inequalities (1)--(6). 
	\end{itemize}   
	This completes the proof of the lemma.    
\end{proof}

\begin{lemma} \label{lem5}
	Let  $i\in[1,n]$, $j\in[1,i)$,    $b=\mathtt{low}(i) $,   $c=\mathtt{low}(j) $,     $d =\mathtt{low}(j+1,i-1) $,     $d'=\mathtt{low}(b,j-1) $,  $d''=\mathtt{low}(b,i-1) $, 
	$z=\max\{x\in (j ,i ):  d=\mathtt{low}(x) \}$, 
	and $k'=\min\{k,b,d,d'\}$.

	\begin{itemize}
		\item[(i)]
		If   either $j<b$  and    $j< \mathtt{low} (  x )$  for some $x\in (j, b)$  or       $k<\min\{ b,c\} $  and    $k\leq  \mathtt{low} (  x )$  for some $x\in [k, \min\{ b,c\} )$,  then  $\gamma^1_{[1,2]} ( k,i:i,j)$ is not defined,
		\item[(ii)]   assume that    $j<b$  and  $z<b$, 
		
		\begin{itemize}
			\item[(a)]  
			if       $k=j  \leq d$,  then   $\gamma^1_{[1,2]} ( k,i:i,j)=\gamma^1_{[1,2]} (   j)+1$,
			\item[(b)]  if   $k<j $  and  $\min\{k,d''\} \leq d$, 
			then   $\gamma^1_{[1,2]} ( k,i:i,j)=\gamma^1_{[1,2]} ( \min\{k,d''\},j: j)+1$, 
			\item[(c)]  if    $ d<\min\{k,d''\}$, then   
			$\gamma^1_{[1,2]} ( k,i:i,j)=\min\{\gamma^1_{[1,2]} ( d,j: j),\gamma^1_{[1,2]} ( d,j: j,x):x\in [d,\min\{k,d''\}) \}+1$,
		\end{itemize} 
		\item[(iii)]  assume that   $j<b$ and  $z\geq b$,
		\begin{itemize}
			\item[(a)]
			
			if      $k=j  \leq d$,  then   $\gamma^1_{[1,2]} ( k,i:i,j)=\gamma^1_{[1,2]} (   j)+1$,
			\item[(b)]  if      $k<j$ and $k\leq d$,  then   $\gamma^1_{[1,2]} ( k,i:i,j)=\gamma^1_{[1,2]} ( k,j: j)+1$,

			\item[(c)]  if    $ d<k$, then   
			$\gamma^1_{[1,2]} ( k,i:i,j)= \gamma^1_{[1,2]} ( d,j: j) +1$,  
			
		\end{itemize}

		\item[(iv)]  
		assume   $j \geq b$, 
		\begin{itemize}
			\item[(a)]  if      $  b=k =j  \leq d $,  then  $\gamma^1_{[1,2]} ( k,i:i,j)= \gamma^1_{[1,2]} (   j) +1$, 
			
			\item[(b)]
			if $b\leq c$, then  $   \gamma^1_{[1,2]} (k, i:i )<\gamma^1_{[1,2]} ( k,i:i,j)$, 
			
			\item[(c)]  if $c<b$, then   
			$\gamma^1_{[1,2]} ( k,i:i,j)= \gamma^1_{[1,2]} ( k',j: j) +1$.
			
			%  \item[(d)]  if  $c<b$  and       $c<k'$,  then   $\gamma^1_{[1,2]} ( k,i:i,j)= \min\{\gamma^1_{[1,2]} ( c,j: j),\gamma^1_{[1,2]} ( c,j: j,x),\gamma^1_{[1,2]} ( c,j: j,x,y) ,    \gamma^{11}_{[1,2]} ( c,j: j,x,y)    :x,y\in [c,k') ,y<x\}+1$.   
		\end{itemize}
		
	\end{itemize}

\end{lemma}

\begin{proof}
	Let $D$ be an    $[1,2]$-dominating  set   of   $G[ 1,i] $ with minimum cardinality  such that $i,j\in D$   and  $x\notin D$ for all $x\in [ k,i) \setminus\{j\}$. So, $|D|=\gamma_{[1,2]}^1(k,i:i,j)$. 
	If      either $j<b$  and    $j< \mathtt{low} (  x )$  for some $x\in (j, b)$, see Figure~\ref{fig16}(a),   or       $k<\min\{ b,c\} $  and    $k\leq  \mathtt{low} (  x )$  for some $x\in [k, \min\{ b,c\} )$,  then   there is no vertex of $D$ to dominate   $x$, a contradiction. 
	This proves  the         case (i).   
	In the rest of the proof, we assume that    the case  (i)    does  not hold. 
	
	Recall that  $d =\mathtt{low}(j+1,i-1) $   and $z=\max\{x\in (j,i):  d=\mathtt{low}(x) \}$.  If $j\geq b$, then  $z\geq b$.   
	Hence, we have   (I)   $j<b$  and  $z<b$,  (II)  $j<b$  and  $z\geq b$,  or (III)   $j\geq b$.   In the following, we consider these cases.  \\

	\begin{figure}[htb]
		\begin{center}
			\includegraphics[width=  \textwidth]{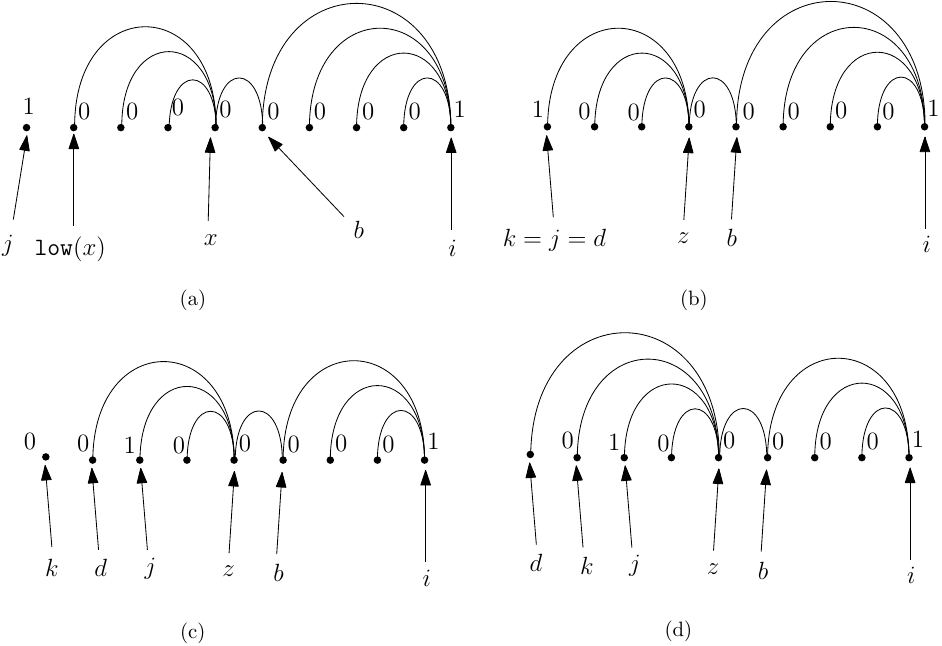}
			\caption{ Illustrating   an    $[1,2]$-dominating  set $D$  of   $G[ 1, i] $    
				such that   $i,j \in D  $  and $x\notin D  $   for all $x\in[k,i)\setminus\{j\}$;  (a)  
				$j<b$  and    $j< \mathtt{low} (  x )$  for some $x\in (j, b)$,       
				(b)    $k=j  \leq d$,   (c)     $k<j$ and $k\leq d$,  and (d)     
				$ d<k$.  
				Note that a vertex with label 1 is in $D$  and  a vertex with label 0  is  not in $D$.} 
			\label{fig16}
		\end{center}
	\end{figure}

	\noindent{\bf Case I.}  
	Assume that     $j<b$   and  $z<b$.  
	If   $d''\leq k$, then 
	$y$   is adjacent to $j$. Hence,  $ y\notin D$  is  adjacent to     vertices $i,j $  and so all vertices in $[d'',j)$ are not  in $D$.
	Since $k\leq j$, we distinguish the following cases:
	\begin{itemize}
		\item [$\bullet$]
		Assume  $k=j  \leq d$.  See Figure~\ref{fig16}(b).  We    get that $D'=D\setminus\{i\}$ is a $[1,2]$-dominating  set   of   $G[ 1,j] $  such that $ j\in D'$.   Hence,  $\gamma^1_{[1,2]} (   j)\leq |D'|=\gamma^1_{[1,2]} ( k,i:i,j)-1$. 
		Conversely, assume that $S$ is    a $[1,2]$-dominating  set   of   $G[ 1,j] $  with minimum cardinality  such that $ j\in S$.   So, $|S|=\gamma^1_{[1,2]} (   j)$.    Since  $j \leq d$ and  the case (i)  does not hold,   each vertex of  $    (j,i) $ is adjacent to at most  two vertices $i$ and $j$. Hence, $S'=S\cup\{i\}$ is an    $[1,2]$-dominating  set   of   $G[ 1,i] $   such that $i, j\in S'$ and $x\notin S'$ for all $x\in    [k,i)\setminus\{j\}$    and so  $\gamma^1_{[1,2]} ( k,i:i,j)\leq |S'|=\gamma^1_{[1,2]} (   j)+1$.   This completes the proof of the case (ii.a).
		\item [$\bullet$]
		Assume   $k<j$.  See Figure~\ref{fig16}(c). 
		Clearly,  $y$   is adjacent to $i$.
		We    get that $D'=D\setminus\{i\}$ is a $[1,2]$-dominating  set   of   $G[ 1,j] $  such that $ j\in D'$  and $x\notin D'$   for all $x\in [\min\{k,d''\},j)$.   Hence,  $\gamma^1_{[1,2]} ( \min\{k,d''\},  j:j)\leq |D'|=\gamma^1_{[1,2]} ( k,i:i,j)-1$. 
		Conversely, assume that $S$ is    a $[1,2]$-dominating  set   of   $G[ 1,j] $  with minimum cardinality  such that $ j\in S$  and $x\notin S$   for all $x\in [\min\{k,d''\},j)$.  So, $|S|=\gamma^1_{[1,2]} ( \min\{k,d''\},  j:j)$. Since  $x\notin S$   for all $x\in [\min\{k,d''\},j)$,   $k \leq d$ and  the case (i)  does not hold,   each vertex of  $  (j,i)$ is adjacent to at most  two vertices $i$ and $j$. Hence, $S'=S\cup\{i\}$ is an    $[1,2]$-dominating  set   of   $G[ 1,i] $   such that $i, j\in S'$ and $x\notin S'$ for all $x\in  [\min\{k,d''\},i)\setminus\{j\}$    and so  $\gamma^1_{[1,2]} ( k,i:i,j)\leq |S'|=\gamma^1_{[1,2]} ( \min\{k,d''\},  j:j)+1$.   This completes the proof of the case (ii.b).
		\item [$\bullet$]
		Assume   $ d<\min\{k,d''\}$.    See Figure~\ref{fig16}(d). Since      $z \notin D$ and $z$ is adjacent to $j$,  the vertex $z$ is  adjacent to at most one vertex  of $D$. Thus, at least one vertex of $[d,\min\{k,d''\}) $ is in $D$,  that is,  either  no vertex of  $[d,\min\{k,d''\})$ is  in $D$   or   exactly one  vertex of  $[d,\min\{k,d''\})$ is  in $D$. In the following, we consider  these cases.

		If      no vertex of  $[d,\min\{k,d''\})$ is  in $D$,  then  we    get that $D'_1=D\setminus\{i\}$ is a $[1,2]$-dominating  set   of   $G[ 1,j] $  such that $ j\in D'_1$  and $x\notin D'_1$   for all $x\in [d,  j )$.  Hence, 
		\begin{eqnarray}\label{eq30}
			\gamma^1_{[1,2]} ( d,  j:j)\leq |D'_1|=\gamma^1_{[1,2]} ( k,i:i,j)-1 . 
		\end{eqnarray}
		
		If      exactly  one  vertex   $w\in [d,\min\{k,d''\}   )$ is  in $D$,  then 
		we    get that $D'_2=D\setminus\{i\}$ is a $[1,2]$-dominating  set   of   $G[ 1,j] $  such that $ j,w\in  D'_2$  and $x\notin D'_2$   for all $x\in [d,  j )\setminus\{w\}$.  Assume that  $x'\in [d,  \min\{k,d''\} ) $.  Hence,  
		\begin{eqnarray}\label{eq31}
			\gamma^1_{[1,2]} ( d,  j:j,x')\leq |D'_2|=\gamma^1_{[1,2]} ( k,i:i,j)-1 .
		\end{eqnarray}
		
		Conversely, assume that $S_1$ is    a $[1,2]$-dominating  set   of   $G[ 1,j] $  with minimum cardinality  such that $ j\in S_1$  and $x\notin S_1$   for all $x\in [d,j)$.  So, $|S_1|=\gamma^1_{[1,2]} ( d,  j:j)$. Since  $x\notin S_1$   for all $x\in [d,j)$,   $d<k$ and  the case (i)  does not hold,   each vertex of  $  (j,i)$ is adjacent to at most  two vertices $i$ and $j$. Hence, $S'_1=S_1\cup\{i\}$ is an    $[1,2]$-dominating  set   of   $G[ 1,i] $   such that $i, j\in S'_1$ and $x\notin S'_1$ for all $x\in  [k,i)\setminus\{j\}$    and so 
		\begin{eqnarray}
			\gamma^1_{[1,2]} ( k,i:i,j)\leq |S'_1|=\gamma^1_{[1,2]} ( d,  j:j)+1.   
		\end{eqnarray}

		Let $x\in [   d, \min\{k,d''\} ) $. Assume that $S_2$ is    a $[1,2]$-dominating  set   of   $G[ 1,j] $  with minimum cardinality  such that $ j,x\in S_2$  and $w\notin S_2$   for all $w\in [d,j)\setminus\{x\}$.   So, $|S_2|=\gamma^1_{[1,2]} ( d,  j:j,x)$. We get that     $S'_2=S_2\cup\{i\}$ is an    $[1,2]$-dominating  set   of   $G[ 1,i] $   such that $i, j\in S'_2$ and $w\notin S'_2$ for all $w\in  [k,i)\setminus\{j\}$    and so 
		\begin{eqnarray}\label{eq35}
			\gamma^1_{[1,2]} ( k,i:i,j)\leq |S'_2|=\gamma^1_{[1,2]} ( d,  j:j,x)+1.   
		\end{eqnarray}
		
		Inequalities (\ref{eq30})--(\ref{eq35}) completes the proof of the case (ii.c).

	\end{itemize}
	
	\noindent{\bf Case II.}  
	Assume that        $j<b$ and  $z\geq b$.
	The proof of cases (iii.a)  and (iii.c)  is similar to the proof of case (ii.a) and   the proof  of case  (iii.b)   is similar to the proof of case (ii.b), respectively.\\

	\noindent{\bf Case III.}  
	Assume that      $j\geq b$. 
	If   $  b=k =j  \leq d $,   then similar to the case (ii.a) we prove that  $\gamma^1_{[1,2]} ( k,i:i,j)= \gamma^1_{[1,2]} (   j) +1$ that completes the proof of the case~(iv.a). 
	Assume that the case  $  b=k =j  \leq d $ does  not hold. 
	We distinguish two  cases  depending on  $b\leq c$  or $c<b$. 
	
	\begin{itemize}
		\item [$\bullet$]  Assume that $b\leq c$.  Recall that $b=\mathtt{low}(i) $ and   $c=\mathtt{low}(j) $.    It obtains that each vertex adjacent to $j$ is also  adjacent to  $i$. We see that  $D'=D\setminus\{j\}$  is an    $[1,2]$-dominating  set   of   $G[ 1,i] $    such that $ i\in D'$ and  $x\notin D'$ for all $x\in [k,i)$ and  so $\gamma^1_{[1,2]}(k,i:i)\leq |D'|= \gamma^1_{[1,2]}(k,i:i,j)-1$. This completes the proof of the case (iv.b).
		\item [$\bullet$] Assume that $c<b$.  Recall  that $k\leq j$,  $d =\mathtt{low}(j+1,i-1) $, $d'=\mathtt{low}(b,j-1) $  and $k'=\min\{k, d,d'\}$. Clearly, $k'$ exists. Since $b\leq j$, $k' \leq  b $.  We claim that   all vertices of  $[k',j) $ are not in $D$. 
		If $k'=k$, the  claim is clear.  
		% Assume that $ k'=b$.        By Lemma~\ref{lem24},  at most one vertex of $D$ is in $[b,i)$.       Since       $j\geq b$  and $j\in D$, we obtain that all vertices of  $[b,j) $ are not in $D$. 
		Assume that $ k'=d$. 
		So, $d\leq b$. 
		Recall that $z=\max\{x\in (j ,i ):  d=\mathtt{low}(x) \}$. We obtain that $z\notin D$ is adjacent to $i,j\in D$. Since $D$ is a $[1,2]$-dominating set,   all vertices of  $[d,j) $ are not in $D$. 
		Assume that  $ k'=d'$.  Let $z'\in[b,j-1]$ such that $    \mathtt{low}(z') =d'$.  
		Since $c<b$, we get that $z'\notin D$ is adjacent to $i,j\in D$ and so 
		all vertices of  $[d',j) $ are not in $D$. 
		Hence,   all vertices of  $[k',j) $ are not in~$D$. 
		The proof of the case (iv.c) is similar to the proof of the case   (ii.b). 
		
	\end{itemize}
	
	\noindent This completes the proof of the lemma.  \end{proof}

Similarly, we can prove the following lemmas.

\begin{lemma} \label{lem6}
	Let    $i\in[1,n]$, $j\in[1,i)$,  $k\in[1,j]$,   $b=\mathtt{low}(i) $,   $c=\mathtt{low}(j) $,  
	$d=\mathtt{low}(k) $,        
	$e =\mathtt{low}(j+1,i-1) $,    
	$e'=\mathtt{low}(b,i-1) $ and   
	$z=\max\{x\in (j ,i ):  e=\mathtt{low}(x) \}$.

	\begin{itemize}
		\item[(i)]
		If     $j<b$  and    $j< \mathtt{low} (  x )$  for some $x\in (j, b)$,       $k<\min\{ b,c\} $  and    $k< \mathtt{low} (  x )$  for some $x\in  (k, \min\{ b,c\} )$  or 
		$l<\min\{ b,c,d\} $  and    $l\leq  \mathtt{low} (  x )$  for some $x\in [l, \min\{ b,c,d\} )$, 
		then  $\gamma^1_{[1,2]} ( l,i:i,j,k)$ is not defined,
		\item[(ii)]   assume that    $j<b$  and  $z<b$, 
		
		\begin{itemize}
			\item[(a)]  
			if       $ j  \leq e$,  then   $\gamma^1_{[1,2]} ( l,i:i,j,k)=\gamma^1_{[1,2]} ( l,  j:j,k)+1$,
			\item[(b)]  if   $e<j $  and  $ e'  \leq k$, 
			then  $\gamma^1_{[1,2]} ( l,i:i,j,k)$ is not defined,
			\item[(c)]   
			if   $e<j $  and  $ e'  > k$, 
			then   $\gamma^1_{[1,2]} ( l,i:i,j,k)=\gamma^1_{[1,2]} ( \min\{l,e\},j: j,k)+1$,  
		\end{itemize} 
		\item[(iii)]  assume that    $z\geq b$,
		\begin{itemize}
			\item[(a)]
			
			if      $e\leq k $,  then   $\gamma^1_{[1,2]} ( l,i:i,j,k) $  is not defined,
			\item[(b)] 
			if      $e>k$,   then   $\gamma^1_{[1,2]} ( l,i:i,j,k)=\gamma^1_{[1,2]} ( l,j: j,k)+1$. 
		\end{itemize}

	\end{itemize}

\end{lemma}

\begin{lemma} \label{lem7}
	Let     $b=\mathtt{low}(i) $,   $c=\mathtt{low}(j) $,  
	$d=\mathtt{low}(k) $,        
	$e =\mathtt{low}(j+1,i-1) $,    
	$e'=\mathtt{low}(b,i-1) $ and   
	$z=\max\{x\in (j ,i ):  e=\mathtt{low}(x) \}$.

	\begin{itemize}
		\item[(i)]
		If    either  $j<b$  and    $j< \mathtt{low} (  x )$  for some $x\in (j, b)$  or        $k<\min\{ b,c\} $  and    $k< \mathtt{low} (  x )$  for some $x\in  (k, \min\{ b,c\} )$,
		then  $\gamma^{11}_{[1,2]} ( l,i:i,j,k)$ is not defined,
		\item[(ii)]   assume that    $j<b$  and  $z<b$, 
		
		\begin{itemize}
			\item[(a)]  if     $e<k$  or  $ e'  \leq k$, 
			then  $\gamma^{11}_{[1,2]} ( l,i:i,j,k)$ is not defined,
			
			\item[(b)]  
			if       $ j  \leq e$ and $k-l=1$,  then   $\gamma^{11}_{[1,2]} ( l,i:i,j,k)=\gamma^1_{[1,2]} ( l,  j:j,k,l)+1$,
			\item[(c)]  
			if       $ j  \leq e$ and $k-l>1$,  then   $\gamma^{11}_{[1,2]} ( l,i:i,j,k)=\gamma^{11}_{[1,2]} ( l,  j:j,k,k-1)+1$,
			\item[(d)]   
			if   $k\leq e<j $,   $ e'  > k$ and $k-l=1$,
			then   $\gamma^{11}_{[1,2]} ( l,i:i,j,k)=\gamma^1_{[1,2]} ( l,j: j,k,l)+1$, 
			\item[(e)]   
			if   $k\leq e<j $,   $ e'  > k$ and $k-l>1$,
			then   $\gamma^{11}_{[1,2]} ( l,i:i,j,k)=\gamma^{11}_{[1,2]} ( l,j: j,k,k-1)+1$, 
		\end{itemize} 
		\item[(iii)]  assume that    $z\geq b$,
		\begin{itemize}
			\item[(a)]
			
			if      $e\leq k $,  then   $\gamma^{11}_{[1,2]} ( l,i:i,j,k) $  is not defined,
			\item[(b)] 
			if      $e>k$  and  $k-l=1$,   then   $\gamma^{11}_{[1,2]} ( l,i:i,j,k)=\gamma^1_{[1,2]} ( l,j: j,k,l)+1$,
			\item[(c)] 
			if      $e>k$  and  $k-l>1$,   then   $\gamma^{11}_{[1,2]} ( l,i:i,j,k)=\gamma^{11}_{[1,2]} ( l,j: j,k,k-1)+1$,  
			
		\end{itemize}

	\end{itemize}

\end{lemma}

Now, we are ready to present our algorithm (Algorithm \ref{alg1}) 
for computing   $\gamma_{[1,2]}(G)$, where $G$ is an  interval graph.
Note that  if $12\in E$, then  the case (iv.b) of   Lemma \ref{lem5}      holds and so      $   \gamma^1_{[1,2]} (1,2:2 )<\gamma^1_{[1,2]} ( 1,2:2,1)$. By considering   Lemmas  \ref{lem2},   \ref{lem3},   \ref{lem14},     \ref{lem4},  \ref{lem5}  and   \ref{lem6}  
we see that  if  we use $\gamma^1_{[1,2]} ( 1,2:2,1)$ for computing a value, then we also use   $   \gamma^1_{[1,2]} (1,2:2 )$ for computing that value. Hence, in Line $12$ we set  $\gamma^1_{[1,2]} ( 1,2:2,1)$ to  $ \infty$.   
We also use this note in Line $21$ for computing  $\gamma^1_{[1,2]} ( k,i:i,j)$.

\renewcommand{\algorithmicrequire}{\textbf{Input:}}
\renewcommand{\algorithmicensure}{\textbf{Output:}}
\begin{algorithm}[h!]
	\caption{\sc  [1,2]-Domination-Interval-Graph$(G )$}
	\label{alg1}
	\algorithmicrequire {An interval graph $G=(V,E)$  of order $n$.}\\
	\algorithmicensure{   $\gamma_{[1,2]}(G)$.}
	\begin{algorithmic}[1]
		\State	Compute a numbering  $(1,2,\ldots , n)$ of vertices of $G$ satisfying the condition  of Corollary~\ref{theo:RR}.
		\State	Compute $\mathtt{ low}(i) $  for all $i\in [1,n]$.
		\State	Compute $\mathtt{ low}( j,i) $ for all $1\leq j<i\leq n$.
		\State	Compute $\mathtt{maxlow}(i) $   for all $i\in [1,n]$.
		\State	$  \gamma^0_{[1,2]}(1)\leftarrow \infty$.
		\State	If   $12\in  E $, then  $  \gamma^0_{[1,2]}(2)\leftarrow 1$, otherwise, $  \gamma^0_{[1,2]}(2)\leftarrow \infty$.
		\State	$  \gamma^0_{[1,2]}(1,2)\leftarrow \infty$. 
		\State	If   $12\in  E $, then  $  \gamma^0_{[1,2]}(1,2:1)\leftarrow 1$, otherwise, $  \gamma^0_{[1,2]}(1,2:1)\leftarrow \infty$.
		\State	$  \gamma^1_{[1,2]}(1)\leftarrow 1$.
		\State	If   $12\in  E $, then  $  \gamma^1_{[1,2]}(2)\leftarrow 1$, otherwise, $  \gamma^1_{[1,2]}( 2)\leftarrow 2$.
		\State	If   $12\in  E $, then   $  \gamma^1_{[1,2]}(1,2:2)\leftarrow 1$, otherwise,  $  \gamma^1_{[1,2]}(1,2:2)\leftarrow  \infty$.
		\State	$  \gamma^1_{[1,2]}(1,2:2,1)\leftarrow    \infty$.
		\State	$  \gamma^{11}_{[1,2]}(1,3:3,2,1)\leftarrow  3$.
		\For{$i=3$  to $n$}	
		\For {$1\leq j < i$}
		\State			Compute $\gamma^0_{[1,2]}(j,i)$ using Lemma \ref{lem2}.
		\State			Compute $\gamma^1_{[1,2]}(j,i:i ) $ using Lemma \ref{lem4}.
		\EndFor
		\For {$1\leq k\leq j< i$}
		\State	Compute $\gamma^0_{[1,2]}(k,i:j)$ using Lemma \ref{lem3}.
		\State	Compute $\gamma^1_{[1,2]}(k,i:i,j)$ using Lemma \ref{lem5}.
		\EndFor
		\For {$1\leq l<k'< j< i$}
		\State		Compute $\gamma^1_{[1,2]}(l,i:i,j,k')$ using Lemma \ref{lem6}.
		\State	Compute $\gamma^{11}_{[1,2]}(l,i:i,j,k')$ using Lemma \ref{lem7}.
		\EndFor
		\State	Compute $\gamma^0_{[1,2]}(i )=\gamma^0_{[1,2]}(i,i)$ using Lemma \ref{lem2}.
		\State	Compute $\gamma^1_{[1,2]}(i ) $ using Lemma \ref{lem14}.
		\EndFor
		\State	\Return    $\min\{\gamma^0_{[1,2]}(n),\gamma^1_{[1,2]}(n )\}$.
	\end{algorithmic}
\end{algorithm}

\begin{theorem}\label{theo1}
	Let $G=(V,E)$ be an interval graph of order $n$. 
	Algorithm \ref{alg1} computes    $\gamma_{[1,2]}(G)$  in $O(n^4)$ time. 
\end{theorem} 

\begin{proof}
	Let    $( 1, 2,\ldots, n)$ be a numbering of vertices of  $G$ 
	computed in Line 1.   This numbering  of vertices  of   $G=(V,E)$  can be computed   in $O(|V|+|E|)$  time~\cite{ramalingam1988unified}. 
	We need  $O(n^3)$ time to 
	compute $\mathtt{ low}(i) $  for all $i\in [1,n]$, 
	$\mathtt{ low}( j,i) $ for all $1\leq j<i\leq n$ and 
	$\mathtt{maxlow}(i) $   for all $i\in [1,n]$ in Lines 2-4. 
	It is easy to compute  the values  
	$  \gamma^0_{[1,2]}(1) $,  $  \gamma^0_{[1,2]}(2) $, 
	$  \gamma^0_{[1,2]}(1,2) $,  
	$  \gamma^0_{[1,2]}(1,2:1) $,
	$  \gamma^1_{[1,2]}(1) $, 
	$  \gamma^1_{[1,2]}(2) $, 
	$  \gamma^1_{[1,2]}(1,2:2) $,
	$  \gamma^1_{[1,2]}(1,2:2,1) $   and 
	$  \gamma^{11}_{[1,2]}(1,3:3,2,1) $
	as seen in Lines 5-13 in $O(1)$ time. 
	Let $i\in [3,n]$, $j\in [1,i)$,  $k\in [1,j]$, $k'\in [2,j)$ and $l\in [1,k')$.         
	By Lemmas  \ref{lem2},   \ref{lem3},   \ref{lem14},     \ref{lem4},  \ref{lem5},  \ref{lem6} and  \ref{lem7},
	compute $\gamma^0_{[1,2]}(i ) $,               
	$\gamma^0_{[1,2]}(j,i)$,
	$\gamma^0_{[1,2]}(k,i:j)$,               
	$\gamma^1_{[1,2]}(i ) $,               
	$\gamma^1_{[1,2]}(j,i:i ) $,  
	$\gamma^1_{[1,2]}(k,i:i,j)$,   
	$\gamma^1_{[1,2]}(k,i:i,j,k')$ and  $\gamma^{11}_{[1,2]}(k,i:i,j,k')$   
	in Lines 15-29 We see that the running time of   Lines 14-29 is $O(n^4)$. 
	Hence, the running time of Algorithm \ref{alg1} is  $O(n^4)$. 
	Clearly,  $\gamma_{[1,2]} (G) =\min\{\gamma^0_{[1,2]}(n),\gamma^1_{[1,2]}(n)\}$. This completes the proof of the theorem. \end{proof}  

\section{$[1,2]$-domination in Circle Graphs}\label{Split Graph}

In this section, we prove hardness results for the $[1,2]$-domination problem in circle graphs. 
We start with the main result of this section.
\\
\begin{theorem}\label{Reduction}
	The $[1,2]$-domination problem is $NP$-complete in circle graphs.
\end{theorem}
\begin{proof} Clearly, the $[1,2]$-domination problem is in $NP$. To prove that the problem is $NP$-hard, We shall reduce (in polynomial time) the $3SAT$ problem to the problem of finding a $[1,2]$-dominating set. Let $X=\{x_1,x_2,\dots,x_n\}$ be the set of boolean variable and $C=\{C_1,C_2, \dots, C_m\}$ be the set of clauses in an arbitrary instance of $3SAT$, like $I$. In the first part, we construct a circle graph $G(I)$ based on the instance $I$, and in the second part, we choose an integer $k$ and show $G(I)$ has a  $[1,2]$-dominating set of size $k$ if and only if $I$ is satisfiable. Let us describe the construction of $G(I)$ produced in the reduction from $3SAT$ to the $[1,2]$-domination problem. Without loss of generality, we assume that every variable appears once in each clause. We choose an arbitrary point of the circle as the origin. The circle graph $G(I)$ is defined as follows:
	
	{\bf First Step: Constructed Gadget for Clause $C_j$.}
	\begin{itemize}
		\item[$\bullet$] We divide the circle into $m$ disjoint open intervals $[s_j, s'_j]$ for $1\leq j\leq m$, corresponding to each clause,
		\begin{figure}[!ht]
			\centering
			\includegraphics[scale=.4]{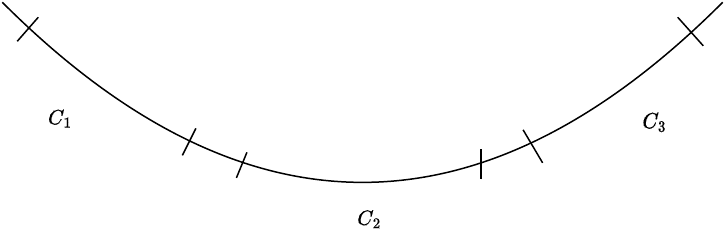}
			\caption{Intervals corresponding to each clause} 
		\end{figure}
		\item[$\bullet$] We create two chords $t_j^l$ and $f_j^l$ associated with each literal $l$, $1\leq l\leq 3$, that appears in clause $C_j$.
		Including the chord $t_j^l$ in a $[1,2]$-dominating set will correspond to setting the associated literal to true and including the $f_j^l$ chord in a $[1,2]$-dominating set will correspond to setting the associated literal to false,
		\item[$\bullet$] For each literal in clause $C_j$, we create a pair of chords $a_j^l$,
		\item[$\bullet$] For each literal in clause $C_j$, we create a pair of chords $p_j^l$. The purpose of each such pair is to dominate them by including the chords $t_j^l$ or $f_j^l$,
		\item[$\bullet$] For each literal in clause $C_j$, the interval between the endpoint of $t_j^l$ and the endpoint of the first chord of $a_j^l$ is called the truth interval, and the interval between the right endpoint of the second chord of $p_j^l$ and the right endpoint of $f_j^l$ is called false interval for literal $l$ of clause $C_j$. (See Figure~\ref{Clause})
		\begin{figure}[!ht]
			\centering
			\includegraphics[scale=.7]{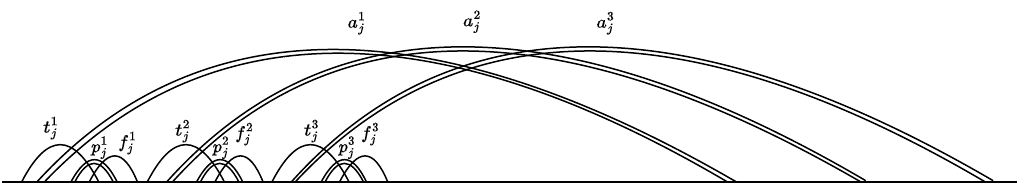}
			\caption{Part of constructed gadget for clause $C_j$} \label{Clause}
		\end{figure}
		
		\item[$\bullet$] For each clause $C_j$, we create three pairs of chords of types $u$ and $w$.
		Based on the endpoints of $a_j^l$ pairs in clause $C_j$, we have pairs $(u_{12},w_{12})$,$(u_{13},w_{13})$ and $(u_{23},w_{23})$. Also, we have two guard chords, $g_1$ and $g_2$, connected to the two claw graphs centered at $c_1$ and $c_2$, respectively. The purpose of each will be discussed later. 
		\begin{itemize}
			\item The chord $u_{13}$ is incident with only the pair $a_j^1$ and $w_{13}$  is incident with only the pair $a_j^3$,
			\item The left point of $u_{12}$ is before pair $a_j^1$ and endpoint of it is before the pair $a_j^3$. So, the chord $u_{12}$  is incident with the pairs $a_j^1$ and $a_j^2$. The chord $w_{12}$ starts before the endpoint $u_{13}$ and after crossing the chord $u_{13}$, the pairs $a_j^2$, and the chord $w_{13}$ will be finished,
			\item The left point of $u_{23}$ is after pair $a_j^1$ and endpoint of it is after the pair $a_j^3$. So, the chord $u_{23}$  is incident with the pairs $a_j^2$ and $a_j^3$. The chord $w_{23}$ starts before the endpoint $u_{13}$ and also before the start point of $w_{12}$. After crossing the chord $u_{13}$, the pairs $a_j^2$, and the chords $w_{12},w_{13}$ will be finished, 
			\item The guard chord $g_1$, from the left endpoint is connected to a chord $c_1$ which three independent chords cross. The purpose of these independent chords is to ensure that $c_1$ is included in any minimum $[1,2]$-dominating set. Therefore, the guard chord $g_1$ is always dominated by $c_1$ once. The right endpoint of $g_1$ will be finished after the start point of $u_{23}$. So, $g_1$  be incident with the chords $ u_{12}, u_{13}$, and $u_{23}$,
			\item The left endpoint of $g_2$ starts before the left endpoint of $u_{12}$ and finishes after the right endpoint of $w_{13}$. Thus, $g_2$ be incidence with the chords $ u_{12}, u_{23}$, and $w_{13}$. The guard chord $g_2$, from the right endpoint, is connected to a chord $c_2$ which is crossed by three independent chords. Similar to chord $g_1$, the purpose of these independent chords is to ensure that $c_2$ is included in any minimum $[1,2]$-dominating set. So, the guard chord $g_2$ is always dominated by $c_2$ once. 
			
		\end{itemize}
		
		(see Fig. \ref{Clause2} for an illustration)
		
		\begin{figure}[!ht]
			\centering
			\includegraphics[scale=.8]{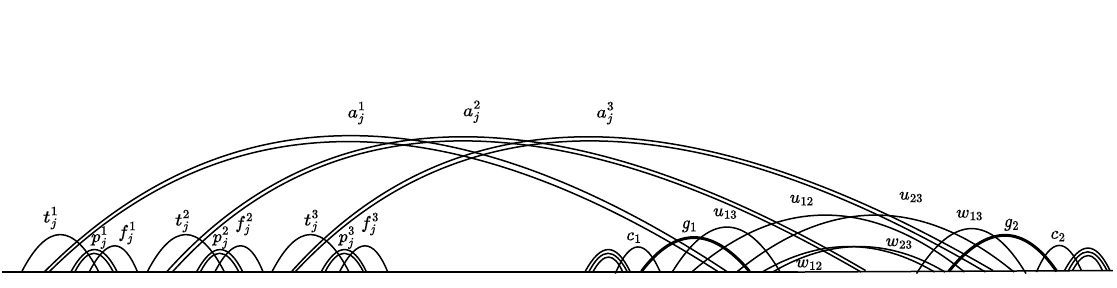}
			\caption{Part of constructed gadget for clause $C_j$} \label{Clause2}
		\end{figure}
	\end{itemize}
	
	The key idea of our proof is as follows. Suppose that $C_j$ is satisfied by some assignment of truth values to the variables. Then we can construct a $[1,2]$-dominating set that contains at least one chord $t_j^l$ corresponding to a true literal in $C_j$ and a pair of $u$ and $w$ chords associated with $C_j$. The chord $t_j^l$ dominates the pairs $a_j^l$ while the other pairs $a_j^l$ and $(u,w)$'s  are dominated by an appropriate pair of $(u ,w)$. To ensure that exactly one of the chords  $t_j^l$  and $f_j^l$ associated with literal $l$ in clause $j$ appears in any $[1,2]$-dominating set, a pair of independent chords $ p_j^l$  associated with literal $l$ is adjacent to both  $t_j^l$  and $f_j^l$ and no other chords. In other words, since the independent chords $p_j^l$'s start from the left endpoint of  $f_j^l$ and end at the right endpoint $t_j^l$, to dominate them, $t_j^l$ or $f_j^l$ or both chords in $p_j^l$ must be select. Later we will argue that any minimum $[1,2]$-dominating set must contain exactly one of the chords  $t_j^l$ or $f_j^l$ otherwise both of $p_j^l$ must be selected, which would lead to a larger dominating set.

	The clause $C_j$ is satisfied if and only if at least one of three literals involved in $C_j$ is true. As we mentioned,  this will correspond to at least one of the chords $t_j^l$ appearing in the dominating set. So this chord will dominate $t_j^l$, $f_j^l$, and both pairs $a_j^l$ and $p_j^l$. To dominate the remaining pairs $a_j$ associated with other literals in $C_j$ we use six chords of type $u$ and $w$ as follows. 
	$u_{13}$ and $w_{13}$ cross the pairs $a_j^1$ and $a_j^3$, respectively. If the second literal is true that leads to $t_j^2$ being in the $[1,2]$-dominating set, two pairs  $a_j^1$ and $a_j^3$ remain undominated. Hence, selecting $u_{13}$ and $w_{13}$ dominates them.
	
	If the first literal is true the pairs  $a_j^2$ and $a_j^3$ are not dominated. Although by choosing $u_{23}$ both pairs $a_j^2$ and $a_j^3$ are dominated, the chords $w_{12}$ and $w_{23}$ remain undominated. As they cross each other, we choose one of them. Similarly, If the third literal is true the pairs  $a_j^1$ and $a_j^2$ are not dominated. Although by choosing $u_{12}$ both pairs $a_j^1$ and $a_j^2$ are dominated, the chords $w_{12}$ and $w_{23}$ remain undominated. As they cross each other, we choose one of them.

	{\bf Second Step: Consistency Condition}
	
	In the above construction, we obtain $m$ pairwise disjoint sections, each associated with one clause. In the second part of our construction, we use connection chords that force a variable to have the same truth value throughout all clauses. This is known as the consistency condition. we create connection chords $tt$, $ff$, $tf$, and $ft$ that connect the clauses as follows. Thus, the consistency of the truth values of literals throughout all clauses will be preserved.
	\begin{itemize}
		\item If variable $x_i$ appears as a positive or negative literal in two clauses $C_j$ and $C_k$, we create a chord $tf$ between two sections associated with clauses $C_j$ and $C_k$ such that the left endpoint of this chord is in the truth interval of literal $x_i$ in clause $C_j$ and its right endpoint is in the false interval of $C_k$. Also, we add a chord $ft$ such that the left endpoint of this chord is in the false interval of literal $x_i$ in the clause $C_j$ and its right endpoint is in the truth interval of $C_k$.
		\item If variable $x_i$ appears as a positive literal in clause $C_j$ and as a negative literal clause $C_k$ or vice versa, we create a chord $tt$ between two sections associated with clauses $C_j$ and $C_k$ such that the left endpoint of this chord is in the truth interval of literal $x_i$ in clause $C_j$ and its right endpoint is in truth interval of $C_k$. Also, we add a chord $ff$ such that the left endpoint of this chord is in the false interval of literal $x_i$ in clause $C_j$ and its right endpoint is in the false interval of $C_k$.
		\item The left endpoints of each connection chord intersect with a claw. In other words, three independent chords connect with another chord to a connection chord.
	\end{itemize}
	Figure \ref{ConnectClauselast} shows all chords associated with two clauses and four variables.
	
	\begin{figure}[!ht]
		\centering
		\includegraphics[scale=.5 ]{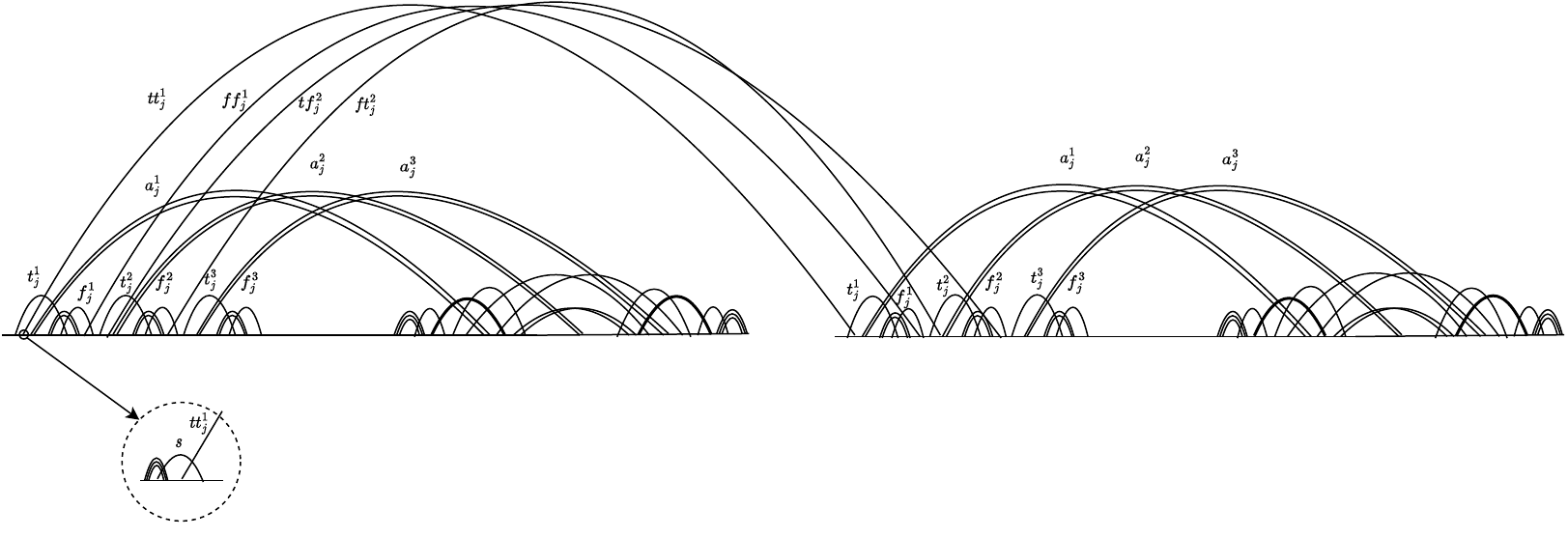}
		\caption{Representation of chords between the clauses $(\bar{x}_1 \vee x_2 \vee x_3)\wedge (x_1 \vee x_2 \vee x_4)$} \label{ConnectClauselast}
	\end{figure}
	
	Let $k = 7m+2t$ where $m$ is the number of clauses and $t$ is the total number of repetitions of variables in clauses.  Theorem \ref{Reduction} directly follows from the following result.

	\begin{lemma}\label{Lemma1}
		$G(I)$ has a $[1,2]$-dominating set $D$ of cardinality at most $k=7m+2t$ if the instance $I$ is satisfiable.
	\end{lemma}
	
	\begin{proof}
		Given a satisfying truth assignment for $I$, we show how to construct a $[1,2]$-dominating set $D$ of size $k=7m+2t$.  Let us consider that the clause $C_j$ contains three literals corresponding to variables $x_i,x_k$, and $x_l$ where $i\leq k \leq l$. We assume the variable $x_i$ appears in the first literal and $x_k,x_l$ are the second and third literals in $C_j$, respectively. Since $I$ is satisfied, at least one of the literals $x_i,x_k$, and $x_l$ must be true. 
		Assume without loss of generality, that the first literal of $C_j$ is true and the second and third literals are false according to the assignment. So we have included $t_j^l,f_j^2$, and $f_j^3$  in $D$. This ensures that all the pairs $p_1,p_2$, and $p_3$ have been dominated. Also, the pair $a_j^1$ is dominated, but the pairs $a_j^2$ and $a_j^3$ are not dominated. To dominate them and all pairs $(u,w)$ in this clause, we put $u_{23},w_{23}$ to $D$. Thus, to dominate this clause, we need at least five chords. Additionally, the chords $c_1$ and $c_2$ must be put in $D$ to dominate the three independent chords crossing them.
		
		Generally, for each literal $l$ in clause $C_j$ that is true in the assignment of $I$, we include the chord(s) $t_j^l$  in $D$  where $l$ is the literal number in clause $C_j$. Otherwise, we include in $D$ the chord(s) $f_j^l$. The chords $t_j^l$ dominate the pair $a_j^l$. To dominate the remaining $a_j$ pairs and the pairs $(u,w)$ in $C_j$, at least one of $(u,w)$ pairs must be chosen. Additionally, the chords $c_1$ and $c_2$ must be included in $D$. 
		
		For each chord $tt,tf,ft$, and $ff$, the chord $s$ connected to them is put in dominating set so all these chords are dominated once. otherwise, we must choose three independent chords connected to $s$, which would lead to $s$ being dominated three times, resulting in a contradiction. Therefore, for each chord $tt, tf, ft$, and $ff$, we need one chord in the dominating set. According to the construction, for each variable repeated in two clauses, we have two connection chords. If the total number of variable repartition is $t$, we include $2t$ chords in $D$. It is clear that $D$ is a $[1,2]$-dominating set.
		
	\end{proof}
	
	\begin{lemma}\label{Lemma2}
		The $3SAT$ instance $I$ is satisfiable if $G(I)$ has a $[1,2]$-dominating set of cardinality at most $k=7m+2t$.
	\end{lemma}
	
	\begin{proof}
		Let $D$ of size $k=7m+2t$ be the given $[1,2]$-dominating set. We show there is a truth assignment to the variables in $X$ that satisfies $I$. Recall that associated with each literal $l$ that appears in a clause $C_j$, there is a pair of chords $p_j^l$. Since no chord is adjacent to two different pairs $p$'s, $D$ must contain, for each $l$ and $j$, either $t_j^l$ or $f_j^l$ or both chords in $p$. If for each literal $l$ in a clause $j$, $D$ contains $t_j^l$ or $f_j^l$(but not both), then exactly $3m$ type $t$ or $f$ chord in $D$ are sufficient to dominate all pairs of chords $p$. Otherwise, if $D$ contains any additional pairs of $p$ more than $3m$ chords are necessary to dominate chords in pairs $p$. Also, all $c_1$'s and $c_2$'s, totally $2m$, chords are in $D$ to dominate the three crossed chords connected to them.
		
		There is no chord in $G(I)$ adjacent to $(u,w)$ pairs associated with different clauses. Thus the set of chords in $D$ which dominate the $(u,w)$ pairs in different clauses are pairwise disjoint. Also, no chord in $G(I)$ is adjacent to all the three pairs $(u,w)$ in a clause. Therefore $D$ must contain at least two chords per clause to dominate all $(u,w)$ pairs in a clause. Combined with the argument in the previous paragraph, this means exactly $3m$ chords in $D$ dominate $p$ pairs and $2m$ chords dominate the $(u,w)$ pairs. This further implies that for each variable $x_i$ appearing in a clause $C_j$, either $t_j^l$ or $f_j^l$ belongs to $D$.
		
		Suppose there are three literal $l_1,l_2$, and $l_3$ in clause $C_j$, and assume none of $t_j^1,t_j^2$, and $t_j^3$ are included in $D$ (i.e., all $f$ chords are selected). In this case, all pairs $a_j^1,a_j^2$, and $a_j^3$ are not dominated. There is no single $(u,w)$ pair that can dominate all three pairs $a_j$'s. So at least one $t$ chord must be selected to dominate one $a_j$ pair, and the other remaining $a_j$'s will be dominated by the appropriate pair $(u,w)$.
		
		Until now, we have shown that for each clause $C_j$, $D$ contains at least one $t$ chord and one pair $(u,w)$. Based on the construction of the gadget, each chord $tt,tf,ft$, and $ff$ are attached to a claw chord, centered at $s$. So, all the chords $s$ must be in domination and as we have $2t$ such chords, then the total size of our $[1,2]$-domination is equal to $|D|=7m+2t$. Now, we can construct a satisfying truth assignment for $I$ as follows. If $D$ contains a $t$ chord associated with a literal involving $x_i$, then we set this literal(and the corresponding value to its variable) to true in $I$. Similarly,  $D$ contains a $f$ chord associated with a literal involving $x_i$, then we set this literal to false in $I$.
	\end{proof}
	
	The proof of theorem \ref{Reduction} follows directly from Lemma \ref{Lemma1} and \ref{Lemma2}.
\end{proof}

\section*{}	
{\bf Conflict of Interests:} The authors declare that they have no known competing financial interests or personal relationships that could have appeared to influence the work reported in this paper.

\section{Conclusion}\label{Conclusion} 
Applying a constraint that determines the maximum number of times a vertex outside the dominating set, like $D$, is dominated by the elements inside $D$, can either make the problem easier or harder than the  domination problem. In this paper, we first designed a polynomial algorithm for the  $[1,2]$-dominating problem in proper interval graphs.  Although $O(n^4)$ might seem large, it provides an affirmative answer to the open problem regarding the complexity class of the $[1,2]$-Domination problem in non-proper interval graphs. Designing a more efficient algorithm for this problem remains an interesting open problem. Next, we have shown that the $[1, 2]$-Dominating Set problem is $NP$-complete on circle graphs. 
It is thus natural to ask for which subclass of circle graphs, like distance-hereditary or k-polygon, the problem has a polynomial-time algorithm. Meanwhile, it would be desirable to show the complexity of $[1,j]$-dominating problem on other classes of graphs marked with a question mark in Figure~\ref{Graphs1}.

\nocite{*}
\bibliographystyle{abbrvnat}
% use the following instead if you encounter problems 

\label{sec:biblio}

\end{document}